\documentclass[11pt]{amsart}
\usepackage{amsmath,amsfonts,amscd,amssymb,epsf}
\newtheorem{theorem}{Theorem}[section]
\newtheorem{proposition}[theorem]{Proposition}
\newtheorem{lemma}[theorem]{Lemma}
\newtheorem{corollary}[theorem]{Corollary}
\theoremstyle{definition}

\theoremstyle{remark} \newtheorem{remark}[theorem]{Remark}
\numberwithin{equation}{section}
\DeclareMathOperator{\id}{id} 
 
 \DeclareMathOperator{\Homeo}{Homeo}

\DeclareMathOperator{\PSL}{PSL} 

\newcommand{\Z}{{\mathbb{Z}}}\newcommand{\Del}{\mathbb{D}}

\newcommand{\C}{{\mathbb{C}}}
\newcommand{\R}{{\mathbb{R}}}
\newcommand{\mL}{\mathcal{L}}
\newcommand{\B}{\mathcal{B}}
\newcommand{\mM}{\mathcal{M}}
\newcommand{\F}{\mathcal{F}}
\newcommand{\mC}{\mathcal{C}}
\newcommand{\bk}{\backslash}
\newcommand{\pa}{\partial}
\newcommand{\Pa}{\mathbb{P}}

\newcommand{\ov}{\overline}

\newcommand{\vep}{\varepsilon}
\newcommand{\z}{\bar{z}}

\begin{document}
\title{Conformal weldings and Dispersionless Toda hierarchy}
\author{ Lee-Peng Teo}\address{Faculty of Information Technology \\
Multimedia University \\ Jalan Multimedia, Cyberjaya\\
63100, Selangor\\Malaysia} \keywords{conformal welding, dispersionless Toda hierarchy, $C^1$ homeomorphisms,  tau function, Grunsky coefficients} \email{lpteo@mmu.edu.my}

\begin{abstract}
Given a $C^1$ homeomorphism of the unit circle $\gamma$, let $f$ and $g$ be respectively the normalized conformal maps from the unit disc and its exterior so that $\gamma=  g^{-1}\circ f$ on the unit circle. In this article, we show that by suitably defined time variables, the evolutions of the pairs $(g, f)$ and $(g^{-1}, f^{-1})$ can be described by an infinite set of nonlinear partial differential equations known as dispersionless Toda hierarchy. Relations to the integrable structure of conformal maps first studied by Wiegmann and Zabrodin \cite{WZ} are discussed. An extension of the hierarchy which contains both our solution and the solution of \cite{WZ} is defined.
\end{abstract}
 \maketitle

\section{Introduction}

Given a $C^1$ homeomorphism of the unit circle $S^1$, the conformal welding or sewing problem looks for a conformal
map $f$ from the unit disc $\Del$  and a conformal map $g$
from the exterior of the unit disc $\Del^*$ so that $\gamma=g^{-1}\circ f$ on the unit circle. We assume that the origin is contained in
$\Omega^+=f(\Del)$ and $\infty$ is contained in $\Omega^-=g(\Del^*)$. Under the conditions $f(0)=0$,  $g(\infty)
=\infty$, $f'(0)g'(\infty)=1$, this problem have a unique solution $(f,
g)$.

Starting from the work of Wiegmann and Zabrodin \cite{WZ}, the
integrable structure of conformal mappings has aroused
considerable interest \cite{KKMWZ, LT, MWZ, 1, 2, 3, 4, 5, 6, 7, 8, 9, 10, 11, 12, 13}. It was found that
the evolution of conformal mappings can be described by an
infinite set of nonlinear partial differential equations, put together
under the name dispersionless Toda hierarchy. Later it was
revealed that it is also closely related to the Dirichlet boundary
problem and two dimensional inverse potential problem \cite{MWZ,
Z}. Putting more precisely, Wiegmann and Zabrodin defined a set
of time variables on the space of analytic curves. They showed
that the deformation of the conformal mapping $g$ of the exterior
domain of an analytic curve with respect to these time variables satisfies
the dispersionless limit of $2$D Toda chain hierarchy. They also
defined the notion of tau function for analytic curves, which is
the tau function for the hierarchy. By a suitable modification, it
was realized that the deformation of the interior conformal
mapping $f$ can also be described by the dispersionless Toda
hierarchy \cite{MWZ}. However, in their approach, the interior
mapping $f$ and exterior mapping $g$ are treated separately, i.e. using
different time coordinates. Dispersionless Toda hierarchy \cite{TT1} can be
interpreted as describing the evolution of the coefficients of two
formal power series $(\mL, \tilde{\mL})$ with respect to a set of
formal time variables $t_n$, $n\in \Z$. Here $\mL(p)= p+$ (lower
power terms), and $\tilde{\mL}(p)= p$+ higher power terms. Hence
under certain analytic conditions, $\mL$ is a function univalent in
a neighborhood of $\infty$, and $\tilde{\mL}$ is a function
univalent in a neighborhood of the origin. In the work of
Wiegmann and Zabrodin, $\mL$ is the power series of $g$ and
$\tilde{\mL}$ is the power series of $1/\bar{g}(z^{-1})$.

In \cite{LT}, Takhtajan put the work of Wiegmann and Zabrodin in
the flavor of conformal field theory. He gave  proofs
to results in \cite{WZ, MWZ, KKMWZ} using the idea of sewing or conformal welding. He introduced the Schottky double $\Pa^1_{\mC}$ of
the Riemann surface $\Omega^-$-- the union of the Riemann surface
$\Omega^-$ and its complex conjugate copy $\ov{\Omega^-}$ gluing
together along their common boundary $\mC$. In this picture
$\ov{\Omega^-}$ is the interior of $\Pa^1_{\mC}$ and
$1/\bar{g}(z^{-1})$ is its Riemann mapping. Hence we can say that
the solution of Wiegmann and Zabrodin \cite{WZ} describes the sewing problem
of $\Pa^1_{\mC}$, where the interior and exterior domains are complex
conjugates of each other.

In this paper, we study the integrable structure of the conformal
welding problem we discussed  at the beginning of the introduction. Let $\Homeo_{C}(S^1)$
denotes the space of $C^1$ homeomorphisms of the unit
circle. Every $C^1$ homeomorphism
$\gamma$ has a unique factorization as $g^{-1}\circ
f\bigr\vert_{S^1}$, where $g$ is a univalent function on the
exterior of the unit disc $\Del^*$, and $f$ is a univalent function on the
unit disc $\Del$, subject to the normalization conditions $f(0)=0$,
$g(\infty) =\infty$ and $f'(0)g'(\infty)=1$.  $f$
and $g$ can be extended to be $C^1$ homeomorphisms of the plane.
Let $\gamma(e^{i\theta}) =\sum_{n\in \Z} c_n
e^{i(n+1)\theta}$ be the Fourier expansion of $\gamma$. We define $t_n$ by taking half a set of the Fourier
coefficients of $\gamma$ and half a set of the corresponding Fourier
coefficients of $1/\gamma^{-1}$. The other half sets of Fourier
coefficients of $\gamma$ and $1/\gamma^{-1}$ are defined as
$v_n$ --- differentiable functions of $t_n$ on $\Homeo_{C}(S^1)$. We
prove in Proposition \ref{Prop2} that the variation of $v_m$ with
respect to $t_n$ is equal to $-|mn| b_{n,m}$, when $m\neq 0, n\neq
0$, and $|m|b_{m,0}$ if $m\neq 0, n=0$, and $-2b_{0,0}$ when
$m=n=0$. Here $b_{n,m}$ are the generalized Grunsky coefficients
of the pair of functions $(f,g)$. A tau function
$\tau$ is constructed and   it is proved that the variation of
$\log \tau$ with respect to $t_n$ is equal to $v_n$. Combining
with Proposition \ref{Prop2}, we get
\begin{align*}
\frac{\pa^2 \F}{\pa t_n\pa t_m} &=-|nm| b_{n,m},
\hspace{0.5cm}n\neq 0, m\neq 0,\\
 \frac{\pa^2 \F}{\pa
t_n \pa t_0} &= |n| b_{n,0}, \hspace{0.5cm}n\neq 0,\hspace{2cm}
\frac{\pa^2\F}{\pa t_0^2} = -2b_{0,0},\nonumber
\end{align*}
where $\F=\log \tau$ is called the free energy. Using our result
in \cite{Teo03}, we conclude that the evolution of the
coefficients of the inverse mappings $g^{-1}$ and $f^{-1}$ with
respect to $t_n$ satisfies the dispersionless Toda hierarchy.
Namely, $(\mL=g^{-1}, \tilde{\mL}=f^{-1})$ is a solution of the
dispersionless Toda hierarchy.

For the pair of mappings $(g, f)$, we define
$t_n$ and $v_n$ as analogues of the $t_n$ and $v_n$ for the inverse mappings $(g^{-1}, f^{-1})$,  replacing $f$ by $f^{-1}$ and $g$ by $g^{-1}$. In this
case we prove that the variation of $v_m$ with respect to $t_n$ is
given by the generalized Grunsky coefficients of the pair
$(f^{-1}, g^{-1})$. We also construct a tau function $\tau$ such
that the variation of $\log \tau$ with respect to $t_n$ is equal
to $v_n$. Hence we also conclude   that $(\mL=g, \tilde{\mL}=f)$
satisfies the dispersionless Toda hierarchy.

In the formal level, our solutions to the dispersionless Toda hierarchy, as well as the solution of Wiegmann and Zabrodin,
can be viewed as the same solution but with $t_n$ interpreted differently. More precisely, although the $v_n$'s and $t_n$'s are
defined differently in each of these solutions, considering $v_m$
as functions of the formal variables $t_n$, they are actually defined
by the same functions. Heuristically, this can be seen from the
fact that the solutions $(\mL, \tilde{\mL})$ in these three
problems all satisfy the Riemann Hilbert data
\begin{align*}
\mL \mM^{-1} = \tilde{\mL},\hspace{3cm} \mM=\tilde{\mM}.
\end{align*}
Here $\mM$ and $\tilde{\mM}$ are the Orlov-Schulman functions. The
equivalent form of this Riemann Hilbert problem has been studied
by Takasaki  \cite{TT3} in connection to string theory. It
implies the following string equation
\begin{align*}
\{ \mL, \tilde{\mL}^{-1}\}=1,
\end{align*}
which characterize our choice of $t_n$ variables over   other choices.

Our solution to the mapping problem and the solution of
Wiegmann and Zabrodin can be viewed as two different phases of the
same problem if we consider a general    problem. Let $\mathfrak{D}$ be the space consists
of all pairs of univalent functions $(f,g)$ satisfying the  conditions $f(0)=0$,
$g(\infty)=\infty$, $f'(0)g'(\infty)=1$, $f$ and $g$ can be
extended as $C^1$   homeomorphisms of the plane.  We define the functions
$t_n$ and $v_n$, $n\in \Z$ as a
direct generalization of those defined for $(f, g)$ associated to a $\gamma\in \text{Homeo}_{C}(S^1)$.  $\{t_n, n\in
\Z\}$ can be considered as complex coordinates giving  the space  $\mathfrak{D}$ a complex structure such that the coefficients of $f$ and $g$ are all analytic functions of $t_n$. We show that  the variations of $(g, f)$ with respect
to the variables $t_n$, $n\in \Z$,
satisfy the dispersionless Toda hierarchy.  The
subspace of $\mathfrak{D}$ defined by the condition
$f(z)=1/\bar{g}(z^{-1})$ is the space considered in the solution
of Wiegmann and Zabrodin. The subspace of $\mathfrak{D}$ defined by
the condition $f(S^1)=g(S^1)$ corresponds to the space we consider
for the  conformal welding problem.

The content of the paper is the following. In Section 2 we
summarize some facts we need about conformal weldings and univalent functions. In Section 3 we quickly
review the dispersionless Toda hierarchy. In Section 4, we define an infinite set of variables $t_n$ and show that the evolution of the pair of inverse mappings $(g^{-1}, f^{-1})$ satisfies the dispersionless Toda hierarchy. In Section 5, an analogous result is proved for the pair of
mappings $(g, f)$. In Section 6, we discuss the Riemann Hilbert data of our
solutions and compare to the solution of Wiegmann and Zabrodin. In
Section7, we discuss a general problem that contains our solution and the solution of \cite{WZ} as special cases.
In the Appendix, we give an example --- the class of conformal
mappings to the discs, which corresponds to the subgroup of linear
fractional transformations on the unit circle.

\section{Conformal Weldings and Univalent functions}

We review some concepts we need about conformal weldings and univalent functions. For
details, see \cite{Ki, KY, Lehto, Gar, Pom, Duren, Teo03, Teo04}.

Let $\Homeo_{C}(S^1)$ be the space of all $C^1$
homeomorphisms on the unit circle $S^1$. We denote by $\mathbb{D}$
the unit disc and $\mathbb{D}^*$ the exterior of the unit disc. A
$C^1$ homeomorphism $\gamma \in \Homeo_{C}(S^1)$ has the
following properties:

 I. It can be extended to be a $C^1$ map on the
complex plane, which is also denoted by $\gamma$, and satisfies
\begin{align*}
\gamma\left(\frac{1}{\z}\right)=\frac{1}{\ov{\gamma(z)}},
\hspace{2cm}\forall z\in \C.
\end{align*}
Moreover, $\gamma$ is real analytic on $\C \setminus S^1$.

II. There exists two $C^1$  homeomorphisms $f$ and $g$ of the plane, such
that $\gamma= g^{-1}\circ f$, and $f\bigr\vert_{\mathbb{D}}$ and
$g\bigr\vert_{\mathbb{D}^*}$ are the unique univalent functions
satisfying $f(0)=0$, $g(\infty)= \infty$ and
$f'(0)g'(\infty)=1$. The decomposition of $\gamma$ as
$g^{-1}\circ f$ is known as conformal welding or sewing. We can
associate to $\gamma$ the simply connected domain $\Omega^+=f(\mathbb{D}) =
g(\mathbb{D})$, its exterior $\Omega^-$ and their common boundary
$\mathcal{C} = f(S^1)=g(S^1)$, a $C^1$ curve. Such an association
is not one-to-one.

Let $\mathfrak{F}(z) = a_1z +a_2 z^2 + \ldots$ be a function
univalent in a neighborhood of the origin and $\mathfrak{G}(z) =
bz + b_0 + b_1 z^{-1} + \ldots$, $b=a_1^{-1}$ be a function
univalent in a neighborhood of $\infty$. We define the
generalized Grunsky coefficients $b_{m,n}$, $m,n\in\Z$ and Faber
polynomials $P_n$ and $Q_n$ by the following formal power series
expansion:
\begin{align*}
\log \frac{\mathfrak{G}(z)-\mathfrak{G}(\zeta)}{z-\zeta}&=
\log b-\sum_{m=1}^{\infty}\sum_{n=1}^{\infty} b_{mn} z^{-m}\zeta^{-n},\end{align*}\begin{align*}
\log \frac{\mathfrak{G}(z) - \mathfrak{F}(\zeta)}{z} &= \log
b-\sum_{m=1}^{\infty}
\sum_{n=0}^{\infty} b_{m,-n} z^{-m} \zeta^n,\end{align*}\begin{align*}
\log \frac{\mathfrak{F}(z)-\mathfrak{F}(\zeta)}{z-\zeta} &=
-\sum_{m=0}^{\infty}
\sum_{n=0}^{\infty} b_{-m, -n} z^{m} \zeta^n,\end{align*}\begin{align*}
\log \frac{\mathfrak{G}(z) -w}{bz} &= -\sum_{n=1}^{\infty}
\frac{P_n(w)}{n}
z^{-n} ,\end{align*}\begin{align*}
\log \frac{w- \mathfrak{F}(z)}{w} &= \log
\frac{\mathfrak{F}(z)}{a_1 z}-\sum_{n=1}^{\infty} \frac{Q_n(w)}{n}
z^{n};
\end{align*}
and for $m\geq 0$, $n\geq 1$, $b_{-m,n}=b_{n,m}$. By definition,
the Grunsky coefficients are symmetric, i.e., $b_{m,n} =b_{n,m}$
for all $m,n\in\Z$. $P_n(w)$ is a polynomial of degree $n$ in $w$
and $Q_n(w)$ is a polynomial of degree $n$ in $1/w$. More
precisely,
\[
P_n(w) = (\mathfrak{G}^{-1}(w))^n_{\geq 0}, \hspace{1cm}Q_n(w) =
(\mathfrak{F}^{-1}(w))^{-n}_{\leq 0}.\] Here when $S$ is a subset
of integers and $A(w)=\sum_{n} A_nw^n$ is a (formal) power series,
we denote by $(A(w))_{S} = \sum_{n\in S} A_n w^n$.

The functions $\log (\mathfrak{G}(z)/z)$, $P \circ \mathfrak{G}$
and $Q\circ \mathfrak{G}$ are meromorphic in a neighborhood of
$\infty$ and the functions $\log (\mathfrak{F}(z)/z)$, $P_n\circ
\mathfrak{F}$ and $Q_n \circ \mathfrak{F}$ are meromorphic in a
neighborhood of the origin. Their power series expansions are
given by
\begin{align}\label{iden1}
\log\frac{\mathfrak{G}(z)}{z} &= \log b-\sum_{m=1}^{\infty}
b_{m,0} z^{-m},\hspace{1.5cm} \log\frac{\mathfrak{F}(z)}{z} =\log
a_1-\sum_{m=1}^{\infty}
b_{-m, 0} z^{m}\\
P_n (\mathfrak{G}(\zeta)) &= \zeta^n + n\sum_{m=1}^{\infty} b_{nm}
\zeta^{-m},\hspace{1.5cm}P_n (\mathfrak{F}(\zeta)) =
nb_{n,0}+n\sum_{m=1}^{\infty}
b_{n, -m} \zeta^m, \nonumber\\
Q_n(\mathfrak{G}(\zeta)) &= -nb_{-n,0} + n\sum_{m=1}^{\infty}
b_{m,-n} \zeta^{-m},\hspace{0.4cm}Q_n (\mathfrak{F}(\zeta))=
\zeta^{-n} + n\sum_{m=1}^{\infty} b_{-n,-m} \zeta^{m}.\nonumber
\end{align}

If $(\mathfrak{F},\mathfrak{G})$ are univalent on $\mathbb{D}$ and
$\mathbb{D}^*$ respectively, we also have the following expansions
that converge on $\mathbb{D}$ and $\mathbb{D}^*$ respectively:
\begin{align}\label{iden4}
\frac{1}{\mathfrak{F}(z)-w}&=-\frac{1}{w}+\sum_{n=1}^{\infty}
\frac{Q_n'(w)}{n}z^n, \hspace{2cm} z\in \mathbb{D}, w\in \Omega^-\\
\frac{\mathfrak{F}'(\zeta)}{\mathfrak{F}(\zeta)-\mathfrak{F}(z)}&=\frac{1}{\zeta-z}-\sum_{m=0}^{\infty}
\sum_{n=1}^{\infty} nb_{-m,-n} z^m\zeta^{n-1}
,\hspace{1cm} z,\zeta \in \mathbb{D},\nonumber\\
\frac{1}{\mathfrak{G}(z)-w} &=\sum_{n=1}^{\infty}
\frac{P_n'(w)}{n} z^{-n}, \hspace{3cm}z\in \mathbb{D}^*, w\in
\Omega^+,\nonumber\\
\frac{\mathfrak{G}'(\zeta)}{\mathfrak{G}(\zeta)-\mathfrak{G}(z)}
&=\frac{1}{\zeta-z} +\sum_{n=1}^{\infty}\sum_{m=1}^{\infty}
mb_{n,m}z^{-n}\zeta^{-m-1},\hspace{1cm} z,\zeta \in
\mathbb{D}^*\nonumber.
\end{align}
This follows immediately from the definition. On the other hand,
when $\mathfrak{F}$ and $\mathfrak{G}$ are complimentary univalent
functions on $\mathbb{D}$ and $\mathbb{D}^*$, i.e. the exterior of
the domain $\mathfrak{F}(\mathbb{D})$ is the domain
$\mathfrak{G}(\mathbb{D}^*)$, we proved in \cite{LT2} that the
semi-infinite matrix $\mathfrak{B}$ and $\mathfrak{C}$ defined by
$$\mathfrak{B}_{mn}=\sqrt{nm}b_{-m,n},\hspace{1cm}\mathfrak{C}_{mn}
=\sqrt{mn}b_{m,-n},\hspace{1cm}m,n\geq 1$$ are invertible
matrices. It follows immediately that
\begin{lemma}\label{ess}Let $\mathfrak{F}$ and $\mathfrak{G}$ be complimentary
univalent functions on $\mathbb{D}$ and $\mathbb{D}^*$ and let
$P_n$, $Q_n$ be their Faber polynomials. \\
I. If there exist constants $\alpha_n$, $n\geq 1$ such that in a
neighborhood of the origin,
\begin{align*}
\sum_{n=1}^{\infty} \alpha_n P_n'(z)=0,
\end{align*}
then $\alpha_n=0$ for all $n$.\\
II. If there exist constants $\beta_n$, $n\geq 1$ such that in a
neighborhood of $\infty$,
\begin{align*}
\sum_{n=1}^{\infty} \beta_n Q_n'(z)=0,
\end{align*}
then $\beta_n=0$ for all $n$.

\end{lemma}\begin{proof}In a neighborhood of the origin, we have
\begin{align*}
\sum_{n=1}^{\infty} \alpha_n
P_n'(\mathfrak{F}(z))\mathfrak{F'}(z)=\sum_{n=1}^{\infty}
\sum_{m=1}^{\infty} nm\alpha_n b_{n,-m}
z^{m-1}=\sum_{m=1}^{\infty}\sqrt{m}(\mathfrak{B}\boldsymbol{\alpha})_mz^{m-1},
\end{align*}
where $\boldsymbol{\alpha}=(\alpha_1, \sqrt{2}\alpha_2,\ldots,
\sqrt{n}\alpha_n,\ldots)^T$. If $\sum_{n=1}^{\infty} \alpha_n P_n'(z)=0$, then by uniqueness of power series
expansion, we have $\mathfrak{B}\boldsymbol{\alpha}=0$. By
invertibility of $\mathfrak{B}$, we conclude that
$\boldsymbol{\alpha}=0$. This prove statement I. Statement II is
proved similarly.
\end{proof}

\section{Dispersionless Toda hierarchy}
The dispersionless Toda hierarchy is a hierarchy of equations
describing the evolution of the coefficients of a pair of formal
power series $(\mL, \tilde{\mL})$:
\begin{align}\label{series}
\mL(p) = r(\boldsymbol{t})p + \sum_{n=0}^{\infty} u_{n+1}(\boldsymbol{t}) p^{-n},\\
(\tilde{\mL}(p))^{-1} = r(\boldsymbol{t})p^{-1} +
\sum_{n=0}^{\infty} \tilde{u}_{n+1}(\boldsymbol{t})
p^{n}\nonumber.
\end{align}
Here $r(\boldsymbol{t})$, $u_n(\boldsymbol{t})$ are functions of
$t_n$, $n\in \Z$, which we denote collectively by
$\boldsymbol{t}$; $p$ is a formal variable independent of
$\boldsymbol{t}$. The evolution of the coefficients $u_n$ are
encoded in the following Lax equations:
\begin{align}\label{Lax}
\frac{\pa \mL}{\pa t_n} =\{ \B_n, \mL\}_T, \hspace{2cm}
\frac{\pa \mL}{\pa t_{-n}} = \{\tilde{B}_n, \mL\}_T,\\
\frac{\pa \tilde{\mL}}{\pa t_n} =\{ \B_n, \tilde{\mL}\}_T,
\hspace{2cm} \frac{\pa \tilde{\mL}}{\pa t_{-n}} = \{\tilde{B}_n,
\tilde{\mL}\}_T.\nonumber
\end{align}
Here $\{\cdot, \cdot\}_T$ is the Poisson bracket
\begin{align*}
\{ f, g\}_T = p \frac{\pa f}{\pa p} \frac{\pa g}{\pa
t_0}-p\frac{\pa f}{\pa t_0} \frac{\pa g}{\pa p}
\end{align*}
and \begin{align*}
\B_n=(\mL^n)_{>0}+\frac{1}{2}(\mL^n)_0,\hspace{1cm}\tilde{B}_n=(\tilde{\mL}^n)_{<0}
+\frac{1}{2}(\tilde{\mL}^n)_{0}.\end{align*} Proposition 3.1 in
\cite{Teo03} can be reformulated as
\begin{proposition}\label{Hirota} If there exists a function $\F$ of
$\boldsymbol{t}$ such that it generates the generalized Grunsky
coefficients of a pair $(\mathfrak{F}, \mathfrak{G})$ of formal
power series, namely
\begin{align*}
\frac{\pa^2\F(\boldsymbol{t})}{\pa t_m\pa t_n}
=\begin{cases}-|mn|b_{m,n}(\boldsymbol{t}),\hspace{1cm}&\text{if}\;\;m\neq 0, n\neq 0\\
|m| b_{m,0}(\boldsymbol{t}), &\text{if}\;\;m\neq 0,n= 0\\
-2b_{0,0}(\boldsymbol{t}),&\text{if}\;\;m=n=0,
\end{cases}
\end{align*}then the pair of
formal power series $( \mathfrak{G}^{-1}, \mathfrak{F}^{-1})$
satisfies the dispersionless Toda hierarchy. Here $
\mathfrak{G}^{-1}$ and  $\mathfrak{F}^{-1}$ are the inverse
functions of $ \mathfrak{G}$ and $\mathfrak{F}$
respectively.
\end{proposition}

\section{The inverse mappings in conformal welding problem and dispersionless Toda
hierarchy}\label{inv} Let $\gamma=g^{-1}\circ f$ be the conformal
welding associated to a point $\gamma\in \Homeo_{C}(S^1)$ and let
$F$ and $G$ be the inverse functions of $f$ and $g$ respectively.
In this section, we are going to define a set of variables $t_n$,
$n\in\Z$ on the space $\Homeo_C(S^1)$. We prove that the evolution
of the mappings $(G, F)$ with respect to these variables satisfies
the dispersionless Toda hierarchy.

Given $\gamma= g^{-1}\circ f \in \Homeo_C(S^1)$, we denote by
$\Omega^+$ the domain $f(\mathbb{D})=g(\mathbb{D})$, $\Omega^-$
its exterior and $\mathcal{C}=f(S^1)=g(S^1)$ the corresponding
$C^1$ -- curve. Let the power series
expansion of $\left.f\right\vert_{\mathbb{D}}$ and
$\left.g\right\vert_{\mathbb{D}^*}$ be given by
\begin{align*}
f(z) &= \sum_{n=1}^{\infty} a_n z^n = a_1 z+ a_2 z^2+\ldots,\\
g(z) &= bz+\sum_{n=0}^{\infty} b_n z^{-n} = bz+ b_0+b_1 z^{-1}
+\ldots.
\end{align*}
They converge on $\mathbb{D}\cup S^1$ and $\mathbb{D}^*\cup S^1$
respectively.

\subsection{The variables $t_n$ and $v_n$}
Since $\gamma$ is $C^1$ on the unit circle, it has Fourier series
expansion on $S^1$ which converges absolutely. We introduce $t_n,
v_n$, $n< 0$, and $t_0$ as coefficients of its Fourier series
expansion:
\begin{align}\label{eq8_4_1}
\gamma(w) = \sum_{n=1}^{\infty} -nt_{-n} w^{-n+1} + t_0 w
+\sum_{n=1}^{\infty} -v_{-n} w^{n+1} , \hspace{1cm} w=e^{i\theta}.
\end{align}
Next, we consider the  mapping $1/\gamma^{-1}=
(1/f^{-1})\circ g= (1/F)\circ g$. We introduce $t_n , v_n$, $n>0$
as coefficients of its Fourier series expansion on $S^1$:
\begin{align}\label{eq8_4_2}
\frac{1}{\gamma^{-1}(w)} =\sum_{n=1}^{\infty}nt_nw^{n-1} + c_0
w^{-1} + \sum_{n=1}^{\infty} v_n w^{-n-1},\hspace{1cm}
w=e^{i\theta}.
\end{align}
For the coefficient $c_0$, we have
\begin{align*}
c_0 =\frac{1}{2\pi i} \oint_{S^1} \frac{1}{\gamma^{-1}(w)} dw=
\frac{1}{2\pi i} \oint_{S^1} \frac{1}{w}d\gamma(w)=\frac{1}{2\pi
i}\oint_{S^1}\frac{\gamma(w)}{w^2}dw =t_0.
\end{align*}
Finally, we define the function $v_0$ on $\Homeo_C(S^1)$ as
\begin{align*}
v_0 =\frac{1}{2\pi i} \oint_{S^1}\left(\left( \log
\frac{f(w)}{w}\right) \frac{\gamma(w)}{w^2}
-\left(\log\frac{g(w)}{w}\right) \frac{1}{\gamma^{-1}(w)}\right)
dw-\frac{1}{2\pi i} \oint_{\mathcal{C}} \frac{G(z)}{F(z)}
\frac{dz}{z}.
\end{align*}
\begin{remark}
Heuristically, we have
\begin{align*}
v_0 =-\frac{1}{2\pi i} \oint_{S^1}\left((\log w)
\frac{\gamma(w)}{w^2} - \frac{\log w}{\gamma^{-1}(w)}\right) dw
\end{align*}
since heuristically
\begin{align*}
\frac{1}{2\pi i} \oint_{S^1}\left(\left( \log f(w)\right)
\frac{\gamma(w)}{w^2} -\left(\log g(w)\right)
\frac{1}{\gamma^{-1}(w)}\right) dw-\frac{1}{2\pi i}
\oint_{\mathcal{C}} \frac{G(z)}{F(z)} \frac{dz}{z}\\
=-\oint_{\mathcal{C}} d\left(\frac{G(z)}{F(z)} \log z\right)=0,
\end{align*}
if we ignore the multi-valued-ness of $\log$ function.
\end{remark}

To see that $t_n, n\in\Z$ give a complete set of local coordinates for $\Homeo_{C}(S^1)$, we need the following lemma.

\begin{lemma}\label{lemma1}
Let $\gamma_t=g_t^{-1}\circ f_t$ be any deformation of  $C^1$
mappings, then
\begin{align*}
\frac{(\pa  \gamma_t/\pa t)\circ F_t(z)}{F_t(z)^2}F_t'(z) =
\left(\frac{\pa}{\pa t} \frac{1}{\gamma_t^{-1}}\right)\circ G_t(z)
G_t'(z), \hspace{1cm} z\in \mathcal{C}.
\end{align*}
\end{lemma}
\begin{proof}
Differentiating $\gamma_t\circ \gamma_t^{-1}= \text{id}$ with respect to $t$ and $w$, we have by chain rule:
\begin{align*}
\frac{\pa \gamma_t}{\pa t}\circ \gamma_t^{-1}(w) + \frac{\pa
\gamma_t}{\pa w}\circ \gamma_t^{-1}(w) \frac{\pa}{\pa
t}\left(\gamma_t^{-1}\right)(w)=0,\\
\frac{\pa \gamma_t}{\pa w}\circ \gamma_t^{-1}(w) \frac{\pa \gamma_t^{-1}(w)}{\pa w}=1,  \hspace{1cm} w\in S^1.
\end{align*}
The assertion follows by setting $w=G(z)$ and using $\gamma_t^{-1}=F_t\circ g_t$.
\end{proof}

\begin{proposition}\label{Prop1}
%The functions $t_n$, $n\in \Z$ are locally independent over $\R$
%on $\Homeo_C(S^1)$.
If $\gamma_t$, $t\in (-\vep, \vep) \subseteq\R$ is a curve on
$\Homeo_C(S^1)$ such that $dt_n/dt=0$ for all $n$, then
$\gamma_t=\gamma_0$ for all $t$.
\end{proposition}
\begin{proof}
%For the first part, we have for $n\geq 0$,
%\begin{align*}
%\frac{d t_n}{d t} &= \frac{1}{2\pi i n}\oint_{S^1} \left(\frac{d
%}{d t} \left(\frac{1}{\gamma^{-1}}\right)\right)(w) w^{-n} dw
%\\
%&= \frac{1}{2\pi i n}\oint_{\mathcal{C}} \left(\frac{d }{d t}
%\left(\frac{1}{\gamma^{-1}}\right)\right)\circ
%g^{-1}(z)(g^{-1})'(z) (g^{-1}(z))^{-n} dz,\\
%\frac{d t_{-n}}{d t} &= \frac{-1}{2\pi i n}\oint_{S^1}
%\left(\frac{d \gamma}{d t} \right)(w) w^{n-2} dw
%\\
%&=\frac{-1}{2\pi i n}\oint_{\mathcal{C }}\frac{\left(\frac{d
%\gamma}{d t} \right)\circ f^{-1}(z)}{(f^{-1}(z))^2} (f^{-1})'(z)
%(f^{-1}( z))^{n} dz.
%\end{align*}
%Since $(f^{-1}(z))^n = z^n + $(higher power terms), and
%$(g^{-1}(z))^{-n} =(b^{-1}z)^{-n}+$(lower power terms), they form
%a basis of nonconstant holomorphic functions on $\Omega^+$ and
%$\Omega^-$ respectively. Together with Lemma \ref{lemma1}, we
%conclude that the $t_n$'s are locally independent.

If $dt_n/dt\bigr\vert_{t=0}=0$ for all $n$, then \eqref{eq8_4_1} and \eqref{eq8_4_2} give
\begin{align*}
\frac{d \gamma_t}{d t}(w) = -\sum_{n=1}^{\infty} \frac{d v_{-n}}{d
t} w^{n+1}, \hspace{0.5cm}\left(\frac{d }{d t}
\left(\frac{1}{\gamma^{-1}_t}\right)\right)(w) =
\sum_{n=1}^{\infty} \frac{d v_n}{d t} w^{-n-1}\hspace{0.8cm} w\in
S^1.
\end{align*}
Hence $w^{-2} (d \gamma_t/d t)(w) $ is the boundary value of the
holomorphic function $ -\sum_{n=1}^{\infty}(d v_{-n}/d t) z^{n-1}$
on $\mathbb{D}$. Since $F$ is holomorphic on $\Omega^+$,
\[
\frac{(d  \gamma^t/d t)\circ F(z)}{F(z)^2}F'(z)
\]
is the boundary value of the holomorphic function
\begin{align*}
-\sum_{n=1}^{\infty} \frac{d v_{-n}}{d t} F(z)^{n-1} F'(z)
\end{align*}
on $\Omega^+$. On the other hand, $\left(d
\left(1/\gamma^{-1}\right)/dt\right)(w)$ is the boundary value of
the holomorphic function $\sum_{n=1}^{\infty} (d v_n/d t)
z^{-n-1}$ on $\mathbb{D}^*$. Since $G$ is holomorphic on
$\Omega^-$,
\begin{align*}
\left(\frac{d}{d t} \frac{1}{(\gamma^t)^{-1}}\right)\circ G(z)
G'(z)
\end{align*}
is the boundary value of the holomorphic function
\begin{align*}
\sum_{n=1}^{\infty} \frac{d v_n}{d t} G(z)^{-n-1} G'(z)
\end{align*}
on $\Omega^-$. Lemma \ref{lemma1} then implies that we have a
holomorphic function on $\hat{\C}= \C\cup \{\infty\}$ which
vanishes at $\infty$. Hence it must be identically zero.
Consequently, $d\gamma_t/dt=0$ and the assertion follows.

\end{proof}
\subsection{The variations of the functions $v_m$ with respect
to $t_n$}
\begin{proposition}\label{Prop2}
Let $b_{m,n}$ be the generalized Grunsky coefficients of the pair
of univalent functions $(\left.f\right\vert_{\mathbb{D}},\left.
g\right\vert_{\mathbb{D}^*})$. The variation of $v_m$, $m\in \Z,
m\neq 0$ with respect to $t_n$, $n\in \Z$ is given by the
following:
\begin{align*}
\frac{\pa v_m}{\pa t_n} = -|mn| b_{n,m}, \hspace{1cm}n\neq 0,
\hspace{1cm} \text{and} \hspace{1cm} \frac{\pa v_m}{\pa t_0} = |m|
b_{0,m}.
\end{align*}

\end{proposition}
\begin{proof}
We follow almost the same idea as the proof of Proposition
\ref{Prop1}. For $n \geq 1$, we have
\begin{align*}
\frac{\pa \gamma}{\pa t_n} (w) &= -\sum_{m=1}^{\infty} \frac{\pa
v_{-m}}{\pa t_n} w^{m+1},\\
\frac{\pa }{\pa t_n}\left(\frac{1}{\gamma^{-1}}\right) (w)& =
nw^{n-1}+\sum_{m=1}^{\infty} \frac{\pa v_m}{\pa t_n} w^{-m-1},
\hspace{1cm} w\in S^1.
\end{align*}
Hence restricted to $\mathcal{C}$,
\begin{align*}
\frac{(\pa \gamma/\pa t_n)\circ F(z)}{F(z)^2}F'(z)
\end{align*}
is the boundary value of the holomorphic function
\begin{align*}
-\sum_{m=1}^{\infty} \frac{\pa v_{-m}}{\pa t_n} F(z)^{m-1} F'(z)
\end{align*}
on $\Omega^+$, and
\begin{align*}
\frac{\pa }{\pa t_n}\left(\frac{1}{\gamma^{-1}}\right)\circ G(z)
G'(z)
\end{align*}
is the boundary value of the meromorphic function
\begin{align*}
\left(nG(z)^{n-1} + \sum_{m=1}^{\infty} \frac{\pa v_m}{\pa t_n}
G(z)^{-m-1}\right)G'(z)
\end{align*}
on $\Omega^-$, which has a pole of order $n-1$ at $\infty$. Lemma
\ref{lemma1} implies that they combine to define a meromorphic
function on the plane, with a single pole of order $n-1$ at
$\infty$. Hence it is a polynomial of degree $n-1$, which we
denote by $\sum_{m=1}^{n} \alpha_{n,m} z^{m-1}$. Therefore,
\begin{align}\label{iden3}
-\sum_{m=1}^{\infty} \frac{\pa v_{-m}}{\pa t_n}F(z)^{m-1}
F'(z)=&\sum_{m=1}^{n} \alpha_{n,m} z^{m-1}, \; z\in \Omega^+,\\
\left(nG(z)^{n-1} + \sum_{m=1}^{\infty} \frac{\pa v_m}{\pa t_n}
G(z)^{-m-1}\right)G'(z)=&\sum_{m=1}^{n} \alpha_{n,m} z^{m-1},
\;z\in \Omega^-\nonumber.
\end{align}
On $\Omega^-$, the second equation gives
\begin{align}\label{iden2}
G(z)^n -\sum_{m=1}^{\infty}\frac{1}{m} \frac{\pa v_m}{\pa t_n}
G(z)^{-m}=&\sum_{m=1}^{n} \frac{\alpha_{n,m}}{m}
z^{m}+\alpha_{n,0},
\end{align}
where $\alpha_{n,0}$ is an integration constant. Comparing
coefficients, we conclude that
\begin{align*}
\sum_{m=1}^{n} \frac{\alpha_{n,m}}{m}
z^{m}+\alpha_{n,0}=(G(z)^n)_{\geq 0}=P_n(z),
\end{align*}
where $P_n(z)$ is the $n$-th Faber polynomial of $g$. Let
$z=g(\zeta)$ in \eqref{iden2}, and compare to the series expansion
of $P_n(g(\zeta))$ given in \eqref{iden1}, we conclude that for
$m, n\geq 1$,
\begin{align*}
\frac{\pa v_m}{\pa t_n} = -mn b_{nm}.
\end{align*}
Now the first equation in \eqref{iden3} gives for $z\in \Omega^+$,
\begin{align*}
-\sum_{m=1}^{\infty}\frac{1}{m} \frac{\pa v_{-m}}{\pa t_n}
F(z)^{m} =P_n(z)+\tilde{\alpha}_{n,0},
\end{align*}
where $\tilde{\alpha}_{n,0}$ is another integration constant.
Putting $z=f(\zeta)$ and compare to the series expansion of
$P_n(f(\zeta))$ given in \eqref{iden1}, we conclude that for $m,
n\geq 1$,
\begin{align*}
\frac{\pa v_{-m}}{\pa t_n} = -mn b_{n, -m}.
\end{align*}
 Now we consider the differentiation with respect to $t_0$.
 Analogous argument shows that restricted to $\mathcal{C}$,
 \begin{align*}
 \frac{(\pa \gamma/\pa t_0)\circ
F(z)}{F(z)^2}F'(z) =\frac{\pa }{\pa
t_0}\left(\frac{1}{\gamma^{-1}}\right)\circ G(z) G'(z)
 \end{align*}
 is the boundary values of the meromorphic function
 \begin{align*}
\left(\frac{1}{F(z)}-\sum_{m=1}^{\infty} \frac{\pa v_{-m}}{\pa
t_0} F(z)^{m-1}\right) F'(z)
 \end{align*}
 on $\Omega^+$, with a simple pole at the origin; and of the
 holomorphic function
 \begin{align*}
\left(\frac{1}{G(z)} + \sum_{m=1}^{\infty} \frac{\pa v_m}{\pa t_0}
G(z)^{-m-1}\right)G'(z)
\end{align*}
on $\Omega^-$. They combine to give a meromorphic function
on $\hat{\C}$, with a single simple pole at the origin. Hence this
function must be a multiple of the function $1/z$. Looking at the
$z^{-1}$ term of both functions, we conclude that it is $1/z$.
After integration, we have
\begin{align*}
\log F(z)-\sum_{m=1}^{\infty} \frac{1}{m}\frac{\pa
v_{-m}}{\pa t_0} F(z)^{m}=\log z +\alpha_{0}, \hspace{0.5cm}&z\in \Omega^+, \\
\log G(z) - \sum_{m=1}^{\infty}\frac{1}{m} \frac{\pa v_m}{\pa t_0}
G(z)^{-m}=\log z + \tilde{\alpha}_{0}, \hspace{0.5cm}&z\in
\Omega^-.
\end{align*}
Put $z=f(\zeta)$ into the first equation and $z=g(\zeta)$ into the
second equation, and compare with \eqref{iden1}, we find that for
$m\neq 0$,
\begin{align*}
\frac{\pa v_m}{\pa t_0} = |m| b_{0, m}.
\end{align*}
Finally differentiation with respect to $t_{-n}$, $n\geq 1$ shows
that the function which equals to
\begin{align*}
\left(-nF(z)^{-n-1}-\sum_{m=1}^{\infty} \frac{\pa v_{-m}}{\pa
t_{-n}} F(z)^{m-1}\right) F'(z),
\hspace{1cm} z\in \Omega^+,\\
\left(  \sum_{m=1}^{\infty} \frac{\pa v_m}{\pa t_{-n}}
G^{-m-1}\right)G'(z),\hspace{1cm} z\in \Omega^-,
\end{align*}
is a meromorphic function on $\hat{\C}$ with a single pole at the
origin of order $-n-1$. The second equation shows that this
function vanishes to the second order at $\infty$. Hence we can
write this function as $\sum_{m=1}^{n} \beta_{n,m} z^{-m-1}$.
After integration, we get
\begin{align*}
F(z)^{-n}-\sum_{m=1}^{\infty} \frac{1}{m}\frac{\pa v_{-m}}{\pa
t_{-n}} F(z)^{m}=\beta_{n,0}-\sum_{m=1}^{n} \frac{\beta_{n,m}}{m}
z^{-m},
\hspace{1cm} z\in \Omega^+,\\
  -\sum_{m=1}^{\infty} \frac{1}{m}\frac{\pa v_m}{\pa t_{-n}}
G(z)^{-m}=\tilde{\beta}_{n,0}-\sum_{m=1}^{n} \frac{\beta_{n,m}}{m} z^{-m},
\hspace{1cm} z\in \Omega^-
\end{align*}
Comparing powers in the first equation, we conclude that
$\beta_{n,0}-\sum_{m=1}^{n} \beta_{n,m} z^{-m}/m=Q_n(z)$, where
$Q_n(z)$ is the $n$-th Faber polynomial of $f$. Finally comparing
with the series expansion of $Q_n(f(\zeta))$ and $Q_n(g(\zeta))$ in \eqref{iden1}
give the desired result
\begin{align*}
\frac{\pa v_m}{\pa t_{-n}}= -|mn| b_{-n,m},
\end{align*}
for $n\geq 1$ and all $m\neq 0$.
\end{proof}

To study the variation of $v_0$ with respect to $t_n$, $n\in \Z$,
we need the following 'calculus formula'. Given $\gamma_t$ a one
parameter family in $\Homeo_C(S^1)$, we denote by $\mathcal{C}^t$
the associated $C^1$ curves. If $h(z, \z, t)$ is a $C^1$ function
in a neighborhood of $\mathcal{C}^t$, then
\begin{align}\label{calculus}
\frac{d}{dt}\Bigr\vert_{t=0} \left(\oint_{\mathcal{C}^t} h(z,
\z,t) dz\right)=\oint_{\mathcal{C}} \frac{dh}{dt}(z,
\z,t)dz+\frac{\pa h}{\pa \z}(z,\z,t)\left(\frac{\pa \ov{w^t}}{\pa
t} dz-\frac{\pa w^t}{\pa t} d\z\right)\Bigr\vert_{t=0}
\end{align}
Here $w^t(z,\z)$ is a family of $C^1$ functions such that
$w^0(z,\z)=\id$, $w^t(\mathcal{C}^t) = \mathcal{C}$. In the
following, we are mostly dealing with function $h$ such that
$h_{\z}=0$ on $\mathcal{C}$. In this case the second term is
missing.

\begin{proposition}\label{Prop3}
The variation of $v_0$ with respect to $t_n$, $n\in \Z$ is given
by
\begin{align*}
\frac{\pa v_0}{\pa t_n} = |n|b_{n,0}=\frac{\pa v_n}{\pa t_0},
\hspace{0.5cm}n\neq 0, \hspace{1cm} \frac{\pa v_0}{\pa t_0} =
-2\log b=2\log a_1.
\end{align*}
\end{proposition}
\begin{proof}
We can check that for any $t$, the expression
\begin{align*}
\frac{\pa \log f_t}{\pa t}\circ F_t(z) \frac{G_t(z)F_t'(z)}{F_t(z)^2}
-\frac{\pa \log g_t}{\pa t} \circ G_t(z)
\frac{G_t'(z)}{F_t(z)}-\frac{1}{z}\left(\frac{\pa }{\pa
t}\frac{G_t}{F_t}\right)(z)  ,
\end{align*}
vanishes identically on $ \mathcal{C}_t$. Together with
\eqref{calculus}, we have
\begin{align*}
\frac{\pa v_0}{\pa t_n}=&\frac{1}{2\pi i}\oint_{S^1} \left( \log
\frac{f(w)}{w}\right)\frac{(\pa \gamma/\pa
t_n)(w)}{w^2}-\left(\log\frac{g(w)}{w}\right)\left( \frac{\pa}{\pa
t_n}\frac{1}{\gamma^{-1}} \right)(w) dw
\end{align*}
Using the series expansion for each term give the desired result.
\end{proof}

Since the Grunsky coefficients are symmetric, from  Proposition
\ref{Prop2} and Proposition \ref{Prop3}, formally there should
exist a function $\mathcal{F}$ on $\Homeo_C{S^1}$ such that $\pa
\F/\pa t_n = v_n$. We define this function in the next section.

\subsection{Tau function}\label{invtau}

First for $\gamma \in \Homeo_C(S^1)$, we define the following two
functions.
\begin{align*}
\psi(z)&=\sum_{n=1}^{\infty} \frac{v_{-n}}{n} z^n,\hspace{0.5cm}
z\in \mathbb{D},
\hspace{1.5cm}\phi(z)=\sum_{n=1}^{\infty}\frac{v_n}{n} z^{-n},
\hspace{0.5cm} z\in \mathbb{D}^*.
\end{align*}
Since $\gamma$ and $1/\gamma^{-1}$ are $C^1$ functions on $S^1$,
these are holomorphic functions on $\mathbb{D}$ and $\mathbb{D}^*$
respectively. We define the tau function $\tau$  by the following formula
\begin{align*}
4\log \tau = 2t_0 v_0 -t_0^2 &+\frac{1}{2\pi i}\oint_{S^1}
\frac{1}{\gamma^{-1}(w)}\left(w\phi'(w)+
2\phi(w)\right) dw \\
&+\frac{1}{2\pi i}\oint_{S^1} \frac{\gamma(w)}{w^2}\left(w\psi'(w)
- 2\psi(w)\right) dw.
\end{align*}

\begin{remark}
When the sum converges absolutely, the tau function can be written
explicitly as
\begin{align*}
4\log \tau= 2t_0 v_0 -t_0^2 -\sum_{n=1}^{\infty}(n-2)(t_n
v_n+t_{-n}v_{-n}).
\end{align*}
We can compare this to the explicit formula for the tau function
of Wiegmann and Zabrodin \cite{WZ, KKMWZ}. We are going to explain
this coincidence in a latter discussion.
\end{remark}

We want to prove that the free energy $\F=\log \tau$ generates the
variables $v_n$, $n\in\Z$.
First, from Proposition \ref{Prop2} and the identities in
\eqref{iden1}, we have the following variational formulas:
\begin{lemma}\label{lemma2}
The variations of the functions $\psi$ and $\phi$ with respect to
$t_n$, $t_{-n}$, $n\geq 1 $ and $t_0$ are given by
\begin{align*}
\frac{\pa \psi}{\pa t_n}(z)&=-P_n(f(z)) + n b_{n,0}, \hspace{1.3cm}
\frac{\pa \phi}{\pa t_n}(z) =-P_n(g(z))+ z^n,\\
\frac{\pa \psi}{\pa t_0}(z)&= -\log \frac{f(z)}{z}+\log
a_1,\hspace{1.5cm}\frac{\pa
\phi}{\pa t_0}(z) = -\log \frac{g(z)}{z}+\log b,\\
\frac{\pa \psi}{\pa t_{-n}}(z)&= -Q_n(f(z))+
z^{-n},\hspace{1.5cm}\frac{\pa \phi}{\pa
t_{-n}}(z)=-Q_n(g(z))-nb_{-n,0}.
\end{align*}
\end{lemma}

Now we prove our claim.
\begin{proposition}\label{thm1}
The tau function generates the functions $v_n$, namely
\begin{align*}
\frac{\pa \log \tau}{\pa t_n}= v_n.
\end{align*}
for all $n\in \Z$.
\end{proposition}
\begin{proof}
For $n\geq 1$, we have
\begin{align*}
&\frac{1}{2\pi i}\oint_{S^1} \left(\left(\frac{\pa}{\pa
t_n}\frac{1}{\gamma^{-1}}\right)(w)\left(w\phi'(w)+
2\phi(w)\right)  +\frac{(\pa\gamma/\pa
t_n)(w)}{w^2}\left(w\psi'(w) -
2\psi(w)\right)\right) dw\\
=&-(n-2)v_n,
\end{align*}
from their explicit series expansion. On the other hand, from
Lemma \ref{lemma2}, we have
\begin{align*}
&\frac{1}{2\pi i}\oint_{S^1}
\frac{1}{\gamma^{-1}(w)}\left(w\frac{\pa}{\pa w}\left(\frac{\pa
\phi}{\pa t_n}\right)(w)+
2\frac{\pa\phi}{\pa t_n}(w)\right) dw \\
&+\frac{1}{2\pi i}\oint_{S^1}
\frac{\gamma(w)}{w^2}\left(w\frac{\pa}{\pa w}\left(\frac{\pa
\psi}{\pa t_n}\right)(w) - 2\frac{\pa\psi}{\pa t_n}(w)\right) dw\\
=&(n+2)v_n - 2nt_0 b_{n,0} -\frac{1}{2\pi
i}\oint_{\mathcal{C}}\left(\frac{G(z)}{F(z)}P_n'(z) +
2\frac{G'(z)}{F(z)}P_n(z)\right)dz\\
&-\frac{1}{2\pi i}
\oint_{\mathcal{C}}\left(\frac{G(z)}{F(z)}P_n'(z)-2
\frac{G(z)F'(z)}{F(z)^2}P_n(z)\right)dz\\
=&(n+2)v_n -2nt_0b_{n,0}-\frac{1}{2\pi
i}\oint_{\mathcal{C}}d\left(\frac{G(z)}{F(z)}P_n(z)\right)\\=&(n+2)v_n -2nt_0b_{n,0}.
\end{align*}
Together with Proposition \ref{Prop3}, we have
\begin{align*}
4\frac{\pa \log \tau}{\pa t_n} = 2nt_0 b_{n,0} -(n-2)v_n +
(n+2)v_n-2nt_0b_{n,0}=4v_n.
\end{align*}
 The cases where $n=0$ and $n\leq -1$ are proved analogously.
\end{proof}
Combining this proposition with Proposition \ref{Prop2} and Proposition \ref{Prop3}, we have
\begin{align}\label{second}
\frac{\pa^2\F}{\pa t_m\pa t_n}
=\begin{cases}-|mn|b_{m,n},\hspace{1cm}&\text{if}\;\;m\neq 0, n\neq 0\\
|m| b_{m,0}, &\text{if}\;\;m\neq 0,n= 0\\
-2b_{0,0},&\text{if}\;\;m=n=0;
\end{cases}
\end{align}
where $\F=\log \tau$.
By Proposition \ref{Hirota}, we   conclude that
\begin{theorem}\label{th1}
The evolution of the functions
$( G, F)$ with respect to $t_n$ satisfies the dispersionless Toda
hierarchy \eqref{Lax}.\end{theorem}

%Since our proof in \cite{Teo03} is somehow formal, we show here
%one example how we can adapt the proof to our purpose.

%Using chain rule, we have the following:
%\begin{align*}
%\frac{\pa g}{\pa t} \circ G(z) + g'(G(z))\frac{\pa G}{\pa
%t}(z)=0,\hspace{1cm} \frac{\pa g}{\pa t} (z) =g(z) \frac{\pa \log
%g}{\pa t}(z) .
%\end{align*}
%Now from \eqref{iden1} and \eqref{second}, it is easy to see that
%in a neighborhood of $\infty$ where it converges, $$\frac{\pa
%(\log g(z)-\log b)}{\pa t_n}=\frac{\pa (P_n\circ g)(z)}{\pa
%t_0}.$$
% On the other hand, we have
%\begin{align*}
%\frac{\pa \log b}{\pa t_n} = -\frac{1}{2}\frac{\pa^3 \log
%\tau}{\pa t_0^2\pa t_n} =- \frac{1}{2}\frac{\pa }{\pa t_0} (n
%b_{n,0}) = -\frac{1}{2}\frac{\pa}{\pa t_0} (P_n(0)).
%\end{align*}
%Hence if we denote by $B_n(z) =P_n(z) -(1/2)P_n(0)$, we conclude
%by analytic continuation that on $\mathbb{D}^*$,
%\begin{align*}
%\frac{\pa \log g (z)}{\pa t_n} =\frac{\pa B_n\circ g(z)}{\pa
%t_0}(z)=\frac{\pa B_n}{\pa t_0}\circ g(z)+ B_n'(g(z))\frac{\pa
%g(z)}{\pa t_0} .
%\end{align*}
%This gives
%\begin{align*}
%\frac{\pa G(z)}{\pa t_n} =& -zG'(z) \frac{\pa \log
%g}{\pa t_n}\circ G(z) \\
%=&-zG'(z)\left( \frac{\pa B_n(z)}{\pa t_0} -\frac{B_n'(z)}{G'(z)}
%\frac{\pa G(z)}{\pa
%t_0}\right)\\
%=& z\frac{\pa B_n(z)}{\pa z}\frac{\pa G(z)}{\pa t_0} -z\frac{\pa
%B_n(z)}{\pa t_0}\frac{\pa G(z)}{\pa z},
%\end{align*}
%which is the first set of equations of the Lax equation for
%dispersionless Toda hierarchy.

\section{The conformal mappings in conformal welding problem and dispersionless Toda hierarchy}

In this section, we want to prove that the evolution of the
functions $(g,f)$ in the conformal welding with respect to some
suitably defined $t_n$ satisfies the dispersionless Toda hierarchy. The
definitions of $t_n$ mimic what we did in the previous section, by
replacing $f$ with $F$ and $g$ with $G$. However, things become
more involved since we do not have the unit circle $S^1$ where we
can write down convergent series expansion for the functions
involved. Instead we have to work with the curves
$\mathcal{C}$.

\subsection{The variables $t_n$ and $v_n$}
Given $\gamma=g^{-1}\circ f\in \Homeo_C(S^1)$, we construct the
following differentiable functions on $\Homeo_C(S^1)$. We use the
same notation as the previous section since it shall not incur
confusion. For $n \geq 1$, we define
\begin{align*}
t_n &=\frac{1}{2\pi i
n}\oint_{S^1}\frac{(g(w))^{-n}}{f(w)}dg(w)=\frac{1}{2\pi i
n}\oint_{\mathcal{C}} \frac{z^{-n}}{f\circ G(z)} dz, \\
t_{-n} &=\frac{-1}{2\pi in}\oint_{S^1} g(w) (f(w))^{n-2} df(w)
=\frac{-1}{2\pi in}\oint_{\mathcal{C}} g\circ F(z) z^{n-2}
dz, \\
v_n &=\frac{1}{2\pi i
}\oint_{S^1}\frac{(g(w))^n}{f(w)}dg(w)=\frac{1}{2\pi i
}\oint_{\mathcal{C}} \frac{z^{n}}{f\circ G(z)} dz, \\
v_{-n} &=\frac{-1}{2\pi i}\oint_{S^1} g(w) (f(w))^{-n-2} df(w)
=\frac{-1}{2\pi i}\oint_{\mathcal{C}} g\circ F(z) z^{-n-2} dz;
\end{align*}
while for $n=0$,
\begin{align*}
t_0=&\frac{1}{2\pi i }\oint_{S^1}\frac{1}{f(w)}dg(w)=\frac{1}{2\pi
i }\oint_{\mathcal{C}} \frac{1}{f\circ G(z)} dz=\frac{1}{2\pi i
}\oint_{\mathcal{C}} \frac{g\circ F}{z^2} dz, \\
v_0=&\frac{1}{2\pi i} \oint_{\mathcal{C}}\left(\left( \log
\frac{F(z)}{z}\right) \frac{g\circ F(z)}{z^2}
-\left(\log\frac{G(z)}{z}\right) \frac{1}{f\circ
G(z)}\right) dz\\
&-\frac{1}{2\pi i} \oint_{S^1} \frac{g(z)}{f(z)} \frac{dz}{z}.
\end{align*}
Unlike the previous case where the $t_n$ and $v_n$ appear as
Fourier coefficients, we do not understand the significance of the
$t_n$ and $v_n$ here.

We also introduce the following functions:
\begin{align}\label{eq8_5_4}
S_{\pm}(z) =\frac{1}{2\pi i} \oint_{\mC} \frac{(1/f)\circ
G(w)}{w-z} dw, \hspace{1.5cm}\tilde{S}_{\pm}(z)=\frac{1}{2\pi i}
\oint_{\mC} \frac{g\circ F(w)}{w^2(w-z)} dw.
\end{align}
Here $S_+$ and $\tilde{S}_+$ are defined for $z\in\Omega^+$. They
are holomorphic functions on $\Omega^+$. In a neighborhood of the
origin, they have the series expansion
\begin{align}\label{eq8_4_3}
S_+(z) = \sum_{n=1} nt_n z^{n-1}, \hspace{1.5cm} \tilde{S}_{+} (z)
= -\sum_{n=1}^{\infty} v_{-n} z^{n-1}.
\end{align}
$S_{-}$ and $\tilde{S}_{-}$ are defined for $z\in \Omega^-$. In a
neighborhood of infinity, they have the series expansion
\begin{align}\label{eq8_4_4}
S_{-}(z) =-\frac{t_0}{z} -\sum_{n=1}^{\infty} v_n
z^{-n-1},\hspace{1.5cm} \tilde{S}_{-}
(z)=-\frac{t_0}{z}+\sum_{n=1}^{\infty} nt_{-n} z^{-n-1}.
\end{align}
 The theory of
complex analysis tells us that on the curve $\mC$,
\begin{align}\label{eq8_4_5}
(1/f)\circ G(z) = S_+(z)-S_{-}(z), \hspace{1.5cm} \frac{g \circ
F(z)}{z^2} = \tilde{S}_{+}(z)- \tilde{S}_{-}(z).
\end{align}

 To prove that $t_n$, $n\in\Z$ give a complete set of local coordinates on
 $\Homeo_{C}(S^1)$, we need  the analog of Lemma \ref{lemma1}:
\begin{lemma}\label{lemma3}
Given a one-parameter curve $\gamma_t=g_t^{-1}\circ f_t$ on
$\Homeo_{C}(S^1)$, we have
\begin{align*}
\frac{(\pa (g_t\circ f_t^{-1})/\pa t)\circ
f_t(z)}{(f_t(z))^2}f_t'(z) = \left(\frac{\pa}{\pa t}
\frac{1}{f_t\circ g_t^{-1}}\right)\circ g_t(z)
g_t'(z),\hspace{1cm}z\in S^1.
\end{align*}
\end{lemma}
The proof is the same as in Lemma \ref{lemma1}.

\begin{proposition}
%The functions $t_n$'s, $n\in \Z$ are locally independent over $\R$
%on $\Homeo_C(S^1)$.
If $\gamma_t$, $t\in (-\vep, \vep) \subset\R$ is a curve on
$\Homeo_C(S^1)$ such that $d t_n/d t=0$ for all $t$, then
$\gamma_t=\gamma_0$ for all $t$.
\end{proposition}
\begin{proof}
%From the 'calculus formula' \eqref{calculus}, we have for $n \geq
%0$,
%\begin{align*}
% \frac{d t_n}{dt}=&\frac{1}{2\pi
% in}\oint_{\mathcal{C}}\frac{d}{dt}\left(\frac{1}{f\circ g^{-1}}\right)(z)z^{-n}dz
 %\\
 %=&\frac{1}{2\pi
 %in}\oint_{S^1}\left(\frac{d}{d t} \frac{1}{f\circ
%g^{-1}}\right)\circ g(z) g'(z)(g(z))^{-n}dz\\
%\frac{d t_{-n}}{dt} =&\frac{-1}{2\pi i
%n}\oint_{\mC}\left(\frac{d(g\circ f^{-1})}{dt}\right)(z)z^{n-2}
%dz\\ =&\frac{-1}{2\pi
% in}\oint_{S^1} \frac{(\frac{d  }{d t}(g\circ f^{-1}))\circ
%f(z)}{(f(z))^2}f'(z)(f(z))^{n}dz.
%\end{align*}
%Now the same argument as Proposition \ref{Prop1} shows that
%$t_n$'s are independent.

If $\gamma_t$ is a one parameter family in $ \Homeo_C(S^1)$ such
that $dt_n/dt=0$ for all $n$, then \eqref{eq8_4_3} and \eqref{eq8_4_4} give
\begin{align*}
\frac{d S_{+}}{dt}(z)=0,\hspace{2cm} \frac{d
\tilde{S}_{-}}{dt}(z)=0;
\end{align*}
and
$$\frac{dS_-}{dt}=O(z^{-2})\hspace{1cm}\text{as}\;z\rightarrow \infty.$$ Hence we conclude from \eqref{eq8_4_5} that
$d((1/f)\circ G)/dt$ is the boundary value of the holomorphic
function $d S_{-}/dt$ on $\Omega^-$, which vanish at $\infty$.
Since $g$ is holomorphic on $\mathbb{D}^*$,
\begin{align*}
\frac{d}{dt}\left(\frac{1}{f\circ G}\right)\circ g(z)
g'(z),\hspace{1cm} z\in S^1\end{align*} is the boundary value of a
holomorphic function on $\mathbb{D}^*$. Similarly, we conclude
that
\begin{align*}
\frac{\left(d(g\circ F)/dt\right)\circ f(z)}{(f(z))^2}f'(z),
\hspace{1cm} z\in S^1
\end{align*}
is the boundary value of a holomorphic function on $\mathbb{D}$.
Lemma \ref{lemma3} then implies that we have a holomorphic
function on $\hat{\C}$ which vanishes at $\infty$. This function
must vanish identically. Therefore $d(g\circ F)/dt=0$ and working
this formula out explicitly, we have
\begin{align*}
\frac{1}{g'(z)}\frac{dg(z)}{dt}=\frac{1}{f'(z)}\frac{df(z)}{dt}.
\end{align*}
However, restricted to $S^1$,
\begin{align*}
\frac{1}{g'(z)}\frac{dg(z)}{dt}=&\frac{d\log b}{dt} z+ \text{lower
order terms in $z$.}\\
\frac{1}{f'(z)}\frac{df(z)}{dt}=&\frac{d\log a_1}{dt}
z+\text{higher order terms in $z$}.
\end{align*}
Comparing coefficients and using the fact that $a_1=b^{-1}$, we
conclude that $dg/dt=df/dt=0$, and consequently $d\gamma_t/dt=0$.
\end{proof}

\subsection{The variations of the functions $v_m$ with respect to
$t_n$.}  We still denote by $b_{m,n}$ the generalized Grunsky
coefficients of $(f,g)$, and $P_n, Q_n$ their Faber polynomials.
We use $\kappa_{m,n}$, $\mathcal{P}_n, \mathcal{Q}_n$ to
denote the generalized Grunsky coefficients and Faber polynomials
of $(F=f^{-1}, G=g^{-1})$.

\begin{proposition}\label{Prop4}
 For $m \neq 0$, the variation of the function
$v_m$ with respect to the variables $t_n$ is given by the following:
\begin{align*}
\frac{\pa v_m}{\pa t_n}=-|mn|\kappa_{n,m}, \;\; n\neq 0,
\hspace{1cm} \frac{\pa v_m}{\pa t_0}= |m| \kappa_{0,m}.
\end{align*}
\end{proposition}
\begin{proof}
First we consider the variation with respect to $t_n$, $n\geq 1$.
For $w\in S^1$, let
\begin{align}\label{iden6}
\frac{(\pa(g\circ F)/\pa t_n)\circ
f(w)f'(w)}{(f(w))^2}=\frac{\pa}{\pa t_n}\left(\frac{1}{f\circ
G}\right)\circ g(w) g'(w)=\sum_{m\in \Z} \alpha_{n,m} w^m.
\end{align}
From \eqref{eq8_4_4}, we
conclude that $\pa \tilde{S}_{-}/\pa t_n=0$. On the other hand,
using the definition of $\tilde{S}_{-}$, we have on $\Omega^-$
\begin{align*}
0=\frac{\pa \tilde{S}_{-}}{\pa t_n}=\frac{1}{2\pi
i}\oint_{\mC}\frac{\pa (g\circ F)(w)/\pa
t_n}{w^2(w-z)}dw=\frac{1}{2\pi i}\oint_{S^1} \frac{(\pa (g\circ
F)/\pa t_n)\circ f(w)f'(w)}{(f(w))^2(f(w)-z)}dw.
\end{align*}
Using the expansion of $1/(f(w)-z)$ given by \eqref{iden4} to
compute this integral, we get
\begin{align*}
0= -\alpha_{n,-1}z^{-1}+
\sum_{m=1}^{\infty}\frac{\alpha_{n,-m-1}}{m} Q_m'(z)
\end{align*}
for $z$ in a neighborhood of $\infty$. Since $$Q_m'(z)=
O(z^{-2}), \hspace{1cm}\text{as}\;z\rightarrow\infty,$$ we have
$\alpha_{n,-1}=0$. It follows from Lemma \ref{ess} that
$\alpha_{n,m}=0$ for all $m\leq -1$.
For $z \in \Omega^+$, we have by \eqref{eq8_4_3}:
\begin{align*}
nz^{n-1}=\frac{\pa S_+}{\pa t_n}(z)  =\frac{1}{2\pi
i}\oint_{S^1}\frac{(\pa((1/(f\circ G))/\pa t_n)\circ g(w)
g'(w)}{g(w)-z} dw.
\end{align*}
Using the series expansion of $1/(g(w)-z)$ given by \eqref{iden4},
we have for $z$ in a neighborhood of the origin,
\begin{align}\label{iden5}
nz^{n-1} =\sum_{m=1}^{\infty}\frac{\alpha_{n,m-1}}{m}P_m'(z).
\end{align}
On the other hand, using \eqref{iden4} and \eqref{iden6}, we have
for $\zeta\in \mathbb{D}^*$,
\begin{align}\label{eq8_5_1}
\frac{\pa S_{-}}{\pa t_n}\circ g(\zeta)g'(\zeta)=&\frac{-1}{2\pi
i}\oint_{S^1}\frac{\pa}{\pa t_n}\left(\frac{1}{f\circ
g^{-1}}\right)\circ g(w) g'(w)\frac{g'(\zeta)}{g(\zeta)-g(w)} dw\\
=&-\sum_{m=1}^{\infty} \alpha_{n,m-1} \sum_{k=1}^{\infty} kb_{km}
\zeta^{-k-1}=\sum_{m=1}^{\infty} \frac{\alpha_{n,m-1}}{m}
\frac{\pa}{\pa
\zeta}\Bigl(P_m(g(\zeta))-\zeta^m\Bigr)\nonumber\\
=&\frac{\pa}{\pa \zeta}\left( g(\zeta)^n-\sum_{m=1}^{\infty}
\frac{\alpha_{n,m-1}}{m}\zeta^{m}\right).\nonumber
\end{align}
The last equality follows from \eqref{iden5}. Since the power
series expansion at infinity of the left hand side only have
negative power terms, we conclude that
\begin{align}\label{eq8_5_2}
\sum_{m=1}^{\infty}\frac{\alpha_{n,m-1}}{m}\zeta^{m} +c_1=
(g(\zeta))^n_{\geq 0}=\mathcal{P}_n(\zeta) ,
\end{align}
where $c_1$ is an integration constant and $\mathcal{P}_n$ is the
$n$-th Faber polynomial of $G$. This implies that $\alpha_{n,m}=0$
for $m\geq n$.
Putting
$\zeta=G(z)$ in \eqref{eq8_5_1}, integrating, comparing the series expansion at
$\infty$ using \eqref{eq8_4_4} and \eqref{eq8_5_2}, we have
\begin{align*}
\sum_{m=1}^{\infty} \frac{1}{m}\frac{\pa v_m}{\pa t_n} z^{-m}
=z^n- \mathcal{P}_n(G(z))+c_1.
\end{align*}
 From this
 we conclude that for $m,n\geq 1$,
 \begin{align*}
 \frac{\pa v_m}{\pa t_n} = -mn\kappa_{n,m}.
 \end{align*}
Next, consider the holomorphic function on $\mathbb{D}$
 \begin{align*}
\frac{\pa \tilde{S}_+}{\pa t_{n}}\circ f(\zeta)
f'(\zeta)=\frac{-1}{2\pi i}\oint_{S^1}\frac{(\pa (g\circ
f^{-1})/\pa t_n)\circ
f(w)f'(w)}{(f(w))^2}\frac{f'(\zeta)}{f(\zeta)-f(w)} dw.
 \end{align*}
Using \eqref{iden4}, \eqref{iden6} and \eqref{eq8_5_2}, we have
\begin{align*}
\frac{\pa \tilde{S}_+}{\pa t_{n}}\circ f(\zeta)
f'(\zeta)=\sum_{m=0}^{n-1} \alpha_{n,m} \zeta^m
=\mathcal{P}_n'(\zeta).
\end{align*}
Putting $\zeta=F(z)$, integrating, comparing the series
expansion  in a neighborhood of the origin and using \eqref{eq8_4_3}, we find that
\begin{align*}
-\sum_{m=1}^{\infty}\frac{1}{m} \frac{\pa v_{-m}}{\pa t_n} z^m
=\mathcal{P}_{n} (f^{-1}(z))+c_2.
\end{align*}
This implies for $m,n \geq 1$
\[
\frac{\pa v_{-m}}{\pa t_n} = -mn \kappa_{n,-m}.
\]

The variation with respect to $t_{-n}$, $n\geq 1$ follows the same
idea. Let
\begin{align}\label{iden7}
\frac{(\pa (g\circ f^{-1})/\pa t_{-n})\circ
f(w)f'(w)}{(f(w))^2}=\frac{\pa}{\pa t_{-n}}\left(\frac{1}{f\circ
g^{-1}}\right)\circ g(w) g'(w)=\sum_{m\in \Z} \beta_{n,m} w^{m}
\end{align}
on $S^1$. Since $\pa S_+/\pa t_{-n}=0$,
\begin{align*}
\sum_{m=1}^{\infty} \frac{\beta_{n,m-1}}{m} P_{m}'(z)=0.
\end{align*}
By Lemma \ref{ess}, this implies $\beta_{n,m}=0$ for $m\geq 0$.
Next since $\pa \tilde{S}_{-}(z)/\pa t_{-n}=nz^{-n-1}$, we have
\begin{align*}
-\beta_{n,-1}z^{-1} +\sum_{m=1}^{\infty}
\frac{\beta_{n,-m-1}}{m}Q_{m}'(z)=nz^{-n-1}
\end{align*}
on a neighborhood of $\infty$. This implies that $\beta_{n,-1}=0$.
Now for $\zeta\in \mathbb{D}$, we have
\begin{align}\label{eq8_5_3}
&\frac{\pa \tilde{S}_+}{\pa t_{-n}}\circ f(\zeta)
f'(\zeta)=\sum_{m=1}^{\infty} \beta_{n,-m-1} \sum_{k=1}^{\infty}
kb_{-k, -m} \zeta^{k-1}\\
=&\sum_{m=1}^{n} \frac{\beta_{n,-m-1}}{m} \frac{\pa}{\pa
\zeta}\left( Q_m(f(\zeta))-\zeta^{-m}\right)=\frac{\pa}{\pa
\zeta}\left( -(f(\zeta))^{-n}
-\sum_{m=1}^{\infty}\frac{\beta_{n,-m-1}}{m}\zeta^{-m}\right).\nonumber
\end{align}
Since the left hand side does not contain negative powers of
$\zeta$, we conclude that $\beta_{n,m}=0$ for $m\leq -n-2$ and for
some constant $c_3$,
\begin{align*}
\sum_{m=1}^{n}\frac{\beta_{n,-m-1}}{m}\zeta^{-m}+c_3=-(f(\zeta))^{-n}_{\leq
0}=-\mathcal{Q}_n(\zeta),
\end{align*}
where $\mathcal{Q}_n$ is the $n$-th Faber polynomial of $F$. Together with \eqref{eq8_5_3}, this
implies that  in a neighborhood of the origin, we
have
\begin{align*}
-\sum_{m=1}^{\infty} \frac{1}{m}\frac{\pa v_{-m}}{\pa t_{-n}}
z^{m} = \mathcal{Q}_n(F(z)) - z^{-n}+c_3.
\end{align*}
Hence we conclude that for $m,n\geq 1$,
\[
\frac{\pa v_{-m}}{\pa t_{-n}}=-mn \kappa_{-n,-m}.
\]
Finally, for $\zeta\in \mathbb{D}^*$, we have
\begin{align*}
\frac{\pa S_{-}}{\pa t_{-n}}\circ
g(\zeta)g'(\zeta)=-\sum_{m=1}^{n} \beta_{n,-m-1}
\zeta^{-m-1}=-\mathcal{Q}_n'(\zeta).
\end{align*}
Hence  in a neighborhood of $\infty$, we have
\begin{align*}
\sum_{m=1}^{\infty}\frac{1}{m} \frac{\pa v_{m}}{\pa t_{-n}} z^{-m}
= -\mathcal{Q}_n (G(z))+c_4.
\end{align*}
This implies that for $m,n\geq 1$
\begin{align*}
\frac{\pa v_m}{\pa t_{-n}}= -mn \kappa_{-n,m}.
\end{align*}
For variation with respect to $t_0$, let
\begin{align}
\frac{(\pa (g\circ F)/\pa t_{0})\circ
f(w)f'(w)}{(f(w))^2}=\frac{\pa}{\pa t_{0}}\left(\frac{1}{f\circ
G}\right)\circ g(w) g'(w)=\sum_{m\in Z} \chi_m w^m
\end{align}
on $S^1$. Since $\pa S_+(z)/\pa t_0=0$, $\chi_m=0$ for all $m\geq
0$. On the other hand, since $\pa \tilde{S}_{-}/\pa t_0 = -1/z$,
this implies that $$
-\chi_{-1}z^{-1}+\sum_{m=1}^{\infty}\frac{\chi_{-m-1}}{m}Q_m'(z)=-z^{-1}$$
in a neighborhood of infinity. Consequently, $\chi_{-1}=1$ and by
Lemma \ref{ess}, $\chi_m=0$ for all $m\leq -2$. Now using the
  formula \eqref{eq8_5_4} for $S_-$, we find that
\begin{align*}
\frac{\pa S_{-}}{\pa t_{0}}\circ
g(\zeta)g'(\zeta)=-\frac{1}{\zeta},\hspace{1cm}\zeta\in
\mathbb{D}^*.
\end{align*}
Together with \eqref{eq8_4_4}, this implies that in a neighborhood of $\infty$,
we have
\begin{align*}
\sum_{m=1}^{\infty}\frac{1}{m} \frac{\pa v_m}{\pa t_0} z^{-m}
=-\log \frac{G(z)}{z}-\log b.
\end{align*}
Eq. \eqref{iden1} then shows that for $m>0$,
\[
\frac{\pa v_m}{\pa t_0}= m\kappa_{0,m}.
\]
Finally, we have for $\zeta\in \mathbb{D}$,
\begin{align*}
\frac{\pa \tilde{S}_+}{\pa t_0} \circ f(\zeta) f'(\zeta) =
\sum_{k=1}^{\infty} kb_{-k,0} \zeta^{k-1} =-\frac{\pa }{\pa \zeta}
\left(\log \frac{f(\zeta)}{\zeta}\right).
\end{align*}
Together with \eqref{eq8_4_3}, this implies that  in a neighborhood of the origin
\begin{align*}
-\sum_{m=1}^{\infty}\frac{1}{m}\frac{\pa v_{-m}}{\pa t_0} z^m =
\log \frac{F(z)}{z},
\end{align*}
which gives us
\begin{align*}
\frac{\pa v_{-m}}{\pa t_0}= m\kappa_{0,-m}
\end{align*}
for $m>0$. This concludes the proof.

\end{proof}

We gather some facts we  need later from the proof above in
the following corollary.
\begin{corollary}\label{cor1}
We have the following variational formula for the functions
$S_{\pm}$ and $\tilde{S}_{\pm}$. For $n\geq 1$,
\begin{align*}
\frac{\pa S_{+}}{\pa t_n}(z) &=nz^{n-1}, \hspace{3.8cm} \frac{\pa
S_{+}}{\pa t_{-n}} =0,\\
\frac{\pa S_{-}}{\pa t_n}(z) &=-\frac{\pa}{\pa
z}\left(\mathcal{P}_n(G(z))-z^n\right),\hspace{1cm} \frac{\pa
S_{-}}{\pa t_{-n}}(z) =-\frac{\pa}{\pa
z}\left(\mathcal{Q}_n(G(z))\right)\\
\frac{\pa \tilde{S}_+}{\pa t_n}(z)&=\frac{\pa}{\pa
z}\left(\mathcal{P}_n(F(z))\right),\hspace{2.2cm} \frac{\pa
\tilde{S}_+}{\pa t_{-n}}(z)=\frac{\pa}{\pa
z}\left(\mathcal{Q}_n(F(z))-z^{-n}\right)\\
\frac{\pa \tilde{S}_{-}}{\pa t_n}(z)&= 0,\hspace{4.5cm} \frac{\pa
\tilde{S}_{-}}{\pa t_{-n}}(z)= nz^{-n-1}
\end{align*}
For $n=0$,
\begin{align*}
\frac{\pa S_{+}}{\pa t_0}(z)&=0, \hspace{4.8cm}\frac{\pa
S_{-}}{\pa
t_0}(z)=-\frac{G'(z)}{G(z)}, \\
\frac{\pa \tilde{S}_{+}}{\pa t_0}(z)&=\frac{\pa }{\pa z}\log
\frac{F(z)}{z}, \hspace{3cm}\frac{\pa \tilde{S}_{-}}{\pa
t_0}(z)=-\frac{1}{z}.
\end{align*}
\end{corollary}

Now we consider the variation of $v_0$.
\begin{proposition}\label{Prop5}
The variation of $v_0$ with respect to  $t_n$, $n\in\Z$ is given
by
\begin{align*}
\frac{\pa v_0}{\pa t_n}=|n| \kappa_{n,0},
\hspace{0.5cm} m\neq 0, \hspace{1cm}\frac{\pa v_0}{\pa t_0} =2\log
b=-2\log a_1.
\end{align*}
\end{proposition}
\begin{proof}
The proof parallels the proof of Proposition \ref{Prop3}. We
have,
\begin{align*}
\frac{\pa v_0}{\pa t_n}=&\frac{1}{2\pi i}\oint_{\mC} \left( \log
\frac{F(z)}{z}\right)\frac{\pa g\circ F(z)}{\pa
t_n}z^{-2}dz-\left(\log\frac{G(z)}{z}\right)\left(
\frac{\pa}{\pa t_n}\frac{1}{f\circ G} \right)(z) dz\\
=&\frac{1}{2\pi i}\oint_{\mC} \left( \log \frac{F(z)}{z}\right)
\left(\frac{\pa \tilde{S}_+}{\pa t_n}-\frac{\pa \tilde{S}_{-}}{\pa
t_n}\right)
(z) dz\\
&-\frac{1}{2\pi
i}\oint_{\mathcal{C}}\left(\log\frac{G(z)}{z}\right)
\left(\frac{\pa S_+}{\pa t_n}-\frac{\pa S_{-}}{\pa t_n}\right) (z)
dz.
\end{align*}
For $n\geq 1$, since $\log (F(z)/z)$ and $\pa \tilde{S}_+/\pa t_n$
are both holomorphic functions on $\Omega^+$, Cauchy integral
formula implies that the first integral vanishes. Similarly, for
the second integral, we are left with
\begin{align*}
-\frac{1}{2\pi
i}\oint_{\mathcal{C}}\left(\log\frac{G(z)}{z}\right)(nz^{n-1})dz
=n\kappa_{n,0}.
\end{align*}
The cases  $n\leq 0$ are proved similarly.
\end{proof}
Since the Grunsky coefficients are symmetric, Propositions
\ref{Prop4} and \ref{Prop5} implies that there exists a function
$\F$ on $\Homeo_C(S^1)$ such that $\pa\F/\pa t_n=v_n$. We are
going to define this function in the next section.

\subsection{Tau function}
First we define the following two functions. In a neighborhood of
$\infty$, they are given by
\begin{align*}
\Psi(z)=\sum_{n=1}^{\infty}
\frac{v_{-n}}{n}z^{n},\hspace{2cm}\Phi(z) =\sum_{n=1}^{\infty}
\frac{v_n}{n} z^{-n}.
\end{align*}
Since $\Psi'(z)= -\tilde{S}_+(z)$, $\Psi(z)$ can be analytically
continued to a holomorphic function on $\Omega^+$. Analogously,
$\Phi'(z)=S_{-}(z)+t_0/z$ implies that $\Phi(z)$ can be
analytically continued to be a holomorphic function on $\Omega^-$.

The tau function  is defined in the
similar way as in Section 4:
\begin{align*}
4\log \tau = 2t_0 v_0 -t_0^2 &+\frac{1}{2\pi i}\oint_{\mC}
\frac{1}{f\circ G(z)}\left(z\Phi'(z)+
2\Phi(z)\right) dz \\
&+\frac{1}{2\pi i}\oint_{\mC} \frac{g\circ
F(z)}{z^2}\left(z\Psi'(z) - 2\Psi(z)\right) dz.
\end{align*}
We use the same notation $\tau$ for the tau function since it will
not incur any confusion. We are going to show that $\F=\log \tau$
generates the functions $v_n$.

First from Corollary \ref{cor1}, we have \begin{lemma} The
variations of the functions $\Phi$ and $\Psi$ are given by
\begin{align*}
\frac{\pa \Psi}{\pa t_n}(z)&=-\mathcal{P}_n(F(z)) + n
\kappa_{n,0}, \hspace{1.3cm}
\frac{\pa \Phi}{\pa t_n}(z) =-\mathcal{P}_n(G(z))+ z^n,\\
\frac{\pa \Psi}{\pa t_0}(z)&= \log \frac{F(z)}{z}+\log
a_1,\hspace{1.5cm}\frac{\pa
\Phi}{\pa t_0}(z) = -\log \frac{G(z)}{z}-\log b,\\
\frac{\pa \Psi}{\pa t_{-n}}(z)&= -\mathcal{Q}_n(F(z))+
z^{-n},\hspace{1.5cm}\frac{\pa \Phi}{\pa
t_{-n}}(z)=-\mathcal{Q}_n(G(z))-n\kappa_{-n,0}
\end{align*}
\end{lemma}
The integration constants are found by comparing power series
expansion of both sides in a neighborhood of the origin or
$\infty$.

Now we can state our proposition.
\begin{proposition}
The tau function generates the functions $v_n$ on
$\Homeo_{C}(S^1)$, namely
\begin{align*}
\frac{\pa \log\tau}{\pa t_n} = v_n ,
\end{align*}
for all $ n\in \Z$.
\end{proposition}
\begin{proof}
The proof is similar to the proof of Proposition \ref{thm1}.
Instead of using series expansion on $S^1$, we use the Cauchy
integral formula and residue calculus as we do in the proof of
Proposition \ref{Prop5}.
\end{proof}

From this, we conclude as Theorem \ref{th1} that \begin{theorem}
The
evolution of the conformal mappings $(g, f)$ with respect to the
variables $t_n$ satisfies the dispersionless Toda hierarchy.\end{theorem}

\subsection{Symmetry of the  variables under $z\mapsto 1/z$ transformation}
We study the symmetries of the coordinates $t_n$. The
transformation $z\mapsto 1/z$ on the complex plane will
interchange the interior and exterior. The domain $\Omega^+$ that
contains the origin will be mapped to the domain $1/\Omega^+$ that
contains ${\infty}$. If $f$ and $g$ are the Riemann mappings of
$\Omega^+$ and $\Omega^-$ such that $f'(0)g'(\infty)=1$, then
$\tilde{f}(z) =1/g(z^{-1})$ and $\tilde{g}(z) =1/f(z^{-1})$ are
the Riemann mappings of the domains $1/\Omega^-$ and $1/\Omega^+$,
and $\tilde{f}'(0)\tilde{g}'(\infty)=1$.
Observe that
\begin{align*}
t_{n} &= \frac{1}{2\pi in } \oint_{S^1}
\frac{g(w)^{-n}}{f(w)}dg(w) =\frac{1}{2\pi in } \oint_{S^1}
\frac{(g(w^{-1}))^{-n}}{f(w^{-1})}dg(w^{-1}) \\
&=\frac{-1}{2\pi in } \oint_{S^1}
{(\tilde{f}(w))^{n-2}}{\tilde{g}(w)} d\tilde{f}(w)
\end{align*}
and similarly
\begin{align*}
t_{-n}&=\frac{1}{2\pi in } \oint_{S^1}\frac
{(\tilde{g}(w))^{-n}}{\tilde{f}(w)} d\tilde{g}(w),\\
v_n &=\frac{-1}{2\pi i } \oint_{S^1}
{(\tilde{f}(w))^{-n-2}}{\tilde{g}(w)}
d\tilde{f}(w),\hspace{0.5cm}v_{-n} =\frac{1}{2\pi i } \oint_{S^1}
\frac{(\tilde{g}(w))^{n}}{\tilde{f}(w)} d\tilde{g}(w).
\end{align*}
This agrees with our observation in \cite{Teo03} that if $(\mL,
\tilde{\mL})$ is a solution of the dispersionless Toda hierarchy,
then $(\mL', \tilde{\mL}')$, where $\mL'(p) =
1/\tilde{\mL}(p^{-1})$, $\tilde{\mL}'(p) =1/\mL(p^{-1})$, is also
a solution  of the dispersionless Toda hierarchy, when we define the new time variables
as $t_n' = t_{-n}$ and $t_{-n}'=t_n$, $n\geq 1$.  The same kind of
symmetry also appears in the variables $t_n$ in Section \ref{inv}.

\section{Riemann Hilbert data}

In the language of Takesaki and Takebe (see \cite{TT1} and the
references therein), to every solution of the dispersionless Toda
hierarchy, one can associate a Riemann Hilbert data (or called the
twistor data). Namely there exist two pairs of functions $(r,h)$
and $(\tilde{r}, \tilde{h})$ of the variable $p$ and $t_0$ such
that
\begin{align}\label{RH}
\{r, h\}_T &= r, \hspace{3cm} \{\tilde{r}, \tilde{h}\}_T=
\tilde{r} , \\
r(\mL, \mM) &= \tilde{r}(\tilde{\mL}, \tilde{\mM}), \hspace{1.5cm}
h(\mL, \mM) = \tilde{h}(\tilde{\mL}, \tilde{\mM}).\nonumber
\end{align}
Here $\mM$ and $\tilde{\mM}$ are the Orlov-Schulman functions.
They are defined so that they can be written as
\begin{align}\label{orlov}
\mM&= \sum_{n=1}^{\infty} nt_n \mL^{n} + t_0 + \sum_{n=1}^{\infty}
v_n \mL^{-n} ,\\
\tilde{\mM} &=\sum_{n=1}^{\infty} -nt_{-n} \tilde{\mL}^{-n} + t_0
+ \sum_{n=1}^{\infty} -v_{-n} \tilde{\mL}^n;\nonumber
\end{align}
and they form a canonical pair with $\mL$ and $\tilde{\mL}$, i.e.
$\{\mL, \mM\}_T=\mL$ and $\{\tilde{\mL},
\tilde{\mM}\}_T=\tilde{\mL}$. Conversely, Takasaki and Takebe also
showed that if $(\mL, \tilde{\mL})$ are formal power series of the
form \eqref{series}, $\mM ,\tilde{\mM}$ are formal functions of
the form \eqref{orlov}, and there exists $(r, h)$ , $(\tilde{r},
\tilde{h})$ satisfying \eqref{RH}, then $(\mL, \tilde{\mL})$ is a
solution to the dispersionless Toda hierarchy. The proof is
formal. The main technique is comparing powers of $p$. However
nothing is assumed about the convergence of the series. For the
 conformal welding problem we are studying, if we define the functions
$\mM$ and $\tilde{\mM}$ as \eqref{orlov}, setting $\mL=g,
\tilde{\mL}=f$, then from the definition of $t_n$ and $v_n$, we
obtain the following relations:
\begin{align}\label{rel}
\mL^{-1}\mM  =\tilde{\mL}^{-1},
\hspace{3cm}\tilde{\mL}\tilde{\mM}=\mL,
\end{align}
on the unit circle. This is exactly the Riemann Hilbert problem
studied by Takasaki \cite{TT3} in connection to string theory.
Here the superscript $-1$ on functions mean the reciprocal. Since $\left\{\mL, \mM\right\}=\mL$, we have the string equation
\begin{align}\label{eq8_5_8}
\left\{ \mL, \tilde{\mL}^{-1}\right\} = \left\{ \mL, \mL^{-1}\mM\right\}= \mL^{-1}\left\{\mL, \mM\right\}=1.
\end{align}The
equations \eqref{rel} are formally equivalent to
\begin{align}\label{rel2}
\mL \mM^{-1}  =\tilde{\mL}, \hspace{3cm} \mM=\tilde{\mM}.
\end{align}
Hence the associated Riemann Hilbert data can be written as $r(p
,t_0) = pt_0^{-1}$, $h(p, t_0)=t_0$, $\tilde{r}(p, t_0) =p$ and
$\tilde{h}(p, t_0)=t_0$. Since our $f$ and $g$ are $C^1$ functions and have
absolutely convergent power series on the unit circle, the proof
of Takasaki and Takebe can be adapted to our case to show that
formally, $(g, f)$ is a solution of the dispersionless Toda
hierarchy, provided we are allowed to differentiate all power
series term by term formally on $S^1$. However, since $t_n$ are
not just formal variables, our proof has shown that everything is
correct even analytically.

On the other hand, we see that the inverse functions $\mL = g^{-1},
\tilde{\mL} = f^{-1}$ and $\mM,
\tilde{\mM}$ defined by \eqref{orlov} also satisfy the same
relation \eqref{rel}, but on the curve $\mathcal{C}$. Hence
by no means we can apply the proof of Takasaki and Takebe to
conclude that $(g^{-1}, f^{-1})$ is a solution of the
dispersionless Toda hierarchy, since the domain where their power
series converge do not overlap in general. However, we have seen
analytically that they indeed satisfy the dispersionless Toda
hierarchy, and the proof is even simpler than the case of $(g, f)$.

Our interest in the inverse mappings $(g^{-1}, f^{-1})$ comes from the
definition of the variables $t_n$. They are related to the
Fourier coefficients of the homeomorphisms on the unit circle.
However, these variables do not have a natural complex structure.
In fact, $\Homeo_C(S^1)$ is an 'odd' dimensional space since the
homogenous space $S^1\bk\Homeo_C(S^1)$ has a complex structure.
We can choose a slice of $S^1\bk\Homeo_C(S^1)$ in $\Homeo_C(S^1)$
by imposing the condition that the coordinate $t_0$ is real.
However, this condition will not be preserved under the $t_n$
flows. On the other hand, $t_{-n}$ in some sense is the complex
conjugate of $t_n$, since
\[
\ov{\gamma^{-1}}(w)=\frac{1}{\gamma^{-1}(w^{-1})}.
\]
We are going to see in the Appendix that in our choice of
variables $t_n$, the subgroup of linear fractional
transformations in $\Homeo_C(S^1)$ correspond to the three
dimensional subspace defined by $t_n =0$ for all $n\geq 2$.

Here we will also like to compare our integrable structure on
conformal welding with the integrable structure of conformal maps
observed by Wiegmann and Zabrodin \cite{WZ, KKMWZ, MWZ}. In the
approach of Wiegman and Zabrodin, the coordinates $t_n$ are
defined as harmonic moments. Given a domain $\Omega^+$ containing $\infty$ and bounded by
an analytic curve $\mathcal{C}$, define
\begin{align*}
t_n &=\frac{1}{2 \pi i}\oint_{\mathcal{C}} z^{-n} \z dz,
\hspace{1cm} v_n =\frac{1}{2\pi i} \oint_{\mC} z^n \z dz,\hspace{1cm}n\geq 1;\\
t_0&=\frac{1}{2\pi i} \oint_{\mC} \z dz , \hspace{1.8cm}
v_0=\frac{1}{\pi }\iint\limits_{\Omega^+} \log |z| d^2z,
\end{align*}
and for $n\geq 1$, $t_{-n}=-\bar{t}_n$, $v_{-n}=-\bar{v}_n$. Let
$g$ be the Riemann mapping from the exterior disc to the
domain $\Omega^+$ normalized so that $g(\infty)=\infty$,
$g'(\infty) >0$. Let
\begin{align}\label{sol}
\mL(z)=g(z),\hspace{1cm}\text{and}\hspace{1cm}
\tilde{\mL}(z)=\frac{1}{\bar{g}(z^{-1})}.
 \end{align}
 Wiegmann and Zabrodin show
that $(\mL, \tilde{\mL})$ is a solution to the  dispersionless
Toda hierarchy. Their solution also satisfies the relation
\eqref{rel2}. Now the close relations between the three solutions
become apparent. In fact, from the relation \eqref{rel}, the
dependence of $v_m$ on the variables $t_n$ are uniquely
determined. Hence although the $t_n$ and $v_m$ have different
interpretation in each of the approach, considering formally as
functions of $t_n$, they in fact correspond to the same solution
of the hierarchy. It is quite remarkable that we can define
appropriate time variables   on each of these problems, such that they
are solutions of the same Riemann Hilbert problem. The Riemann
Hilbert data \eqref{rel2} is significance for it implies the
  string equation \eqref{eq8_5_8}.
In fact,  considering generalized inverse potential
problem, Zabrodin has shown in \cite{Z} that there are other ways to define
$t_n$ and $v_n$ so that \eqref{sol} is a solution of the
dispersionless Toda hierarchy. They correspond to different
Riemann Hilbert data. Similarly, we can also choose other time
variables so that the evolution of  the interior and exterior
mappings $f$ and $g$ and their inverses satisfy the dispersionless
Toda hierarchy.

The solution of Wiegmann and Zabrodin has the advantage that it
allows the introduction of a complex structure on the space of
analytic curves. More precisely, on each slice with constant
$t_0$, $-t_{-n}$ is defined to be the complex conjugate of $t_n$.
Moreover, as is observed by Takhtajan in \cite{LT}, this gives a
close analogy to conformal field theory. The tau function for
analytic curves is a partition function for the theory of free
bosons on the space of analytic curves, and
$$\left\langle \frac{\pa
}{\pa t_n}, \frac{\pa }{\pa t_m}\right\rangle=\frac{\pa^2\log \tau}{\pa
t_n \pa \bar{t}_m}
$$
defines a Hermitian metric on each constant $t_0$ slice. Another
remarkable feature about the solution of Wiegmann and Zabrodin is
that they are closely related to the Dirichlet boundary value
problem (see \cite{MWZ}).

On the other hand, the solution of Wiegmann and Zabrodin cannot
study the evolutions of the interior and exterior mappings using
the same time variables. They have to introduce different time variables for
the interior mapping problem and exterior mapping problem and these variables
are related by a Legendre transformation (see \cite{MWZ}), which
is similar to the symmetry of coordinates we discuss in Section
5.3. It will be interesting to investigate whether we can choose
suitable time variables such that the solutions to the conformal
welding problem also enjoy some of the good features of the
solution of Wiegmann and Zabrodin. This is considered in the next section.

\section{Extension of the conformal mapping  problem}

In this section, we construct a theory that contains both our
solution to the conformal welding problem and the solution of
Wiegmann--Zabrodin to conformal mapping problem. We define the
following spaces:
\begin{align*}
\mathfrak{S}=\Bigl\{ &f: \mathbb{D} \rightarrow \C \;\text{univalent}
\;\bigr\vert\;
f(z)=a_1 z+a_2z^2+ \ldots; a_1\neq0; \\
&f \;\text{is extendable to a $C^1$   homeomorphism
of the plane}.\Bigr\},\\
\Sigma=\Bigl\{ &g: \mathbb{D}^* \rightarrow \C \;\text{univalent}
\;\bigr\vert\;
g(z)=b z+b_0 + b_1z^{-1}+\ldots; b\neq0; \\
&g \;\text{is extendable to a $C^1$   homeomorphism
of the plane}.\Bigr\},\\
\mathfrak{D}=&\left\{(f,g)\; \bigr\vert\; f\in \mathfrak{S}, g\in \Sigma;\;
f'(0)g'(\infty) = a_1b=1 .\right\}.
\end{align*}
We define the following functions on $\mathfrak{D}$:
\begin{align*}
t_n&=\frac{1}{2\pi i
n}\oint_{S^1}\frac{(g(w))^{-n}}{f(w)}dg(w),\hspace{0.5cm} t_{-n}
=\frac{-1}{2\pi in}\oint_{S^1} g(w) (f(w))^{n-2} df(w),
\\
v_n &=\frac{1}{2\pi i }\oint_{S^1}\frac{(g(w))^n}{f(w)}dg(w),
\hspace{0.5cm} v_{-n} =\frac{-1}{2\pi i}\oint_{S^1} g(w)
(f(w))^{-n-2} df(w), \hspace{0.5cm}n\geq 1,\\
t_0 &=\frac{1}{2\pi i
}\oint_{S^1}\frac{1}{f(w)}dg(w)=\frac{1}{2\pi i} \oint_{S^1}\frac{g(w)}{(f(w))^2} df(w), \\
v_0&=\frac{1}{2\pi i} \oint_{S^1}\left(\log\frac{g(w)}{w}\right)
\frac{1}{f(w)} dg(w)-\frac{1}{2\pi i} \oint_{S^1}\left( \log
\frac{f(w)}{w}\right)
\frac{g(w)}{f(w)^2}df(w) \\
&\hspace{1cm}-\frac{1}{2\pi i} \oint_{S^1} \frac{g(w)}{f(w)}
\frac{dw}{w}.
\end{align*}

As in Section 5, we can prove that $\{t_n ,
\bar{t}_n, n\in \Z\}$ is a complete set of local coordinates on
$\mathfrak{D}$.

Let $\Omega_1^+= f(\mathbb{D})$ and $\Omega_1^-$ its exterior. Let
$\Omega_2^-=g(\mathbb{D}^*)$ and $\Omega_2^+$ its exterior. We
 define the following functions:
\begin{align*}
S_{\pm}(z) =\frac{1}{2\pi i} \oint_{S^1} \frac{(1/f)(w)
g'(w)}{g(w)-z} dw, \hspace{1cm} z\in \Omega_2^{\pm},\\
\tilde{S}_{\pm}(z)=\frac{1}{2\pi i} \oint_{S^1} \frac{g(w)
f'(w)}{f(w)^2(f(w)-z)} dw,\hspace{1cm} z\in \Omega_1^{\pm} .
\end{align*}
They are holomorphic functions in the respective domains. In a
neighborhood of the origin or $\infty$, they have the series
expansion
\begin{align*}
S_+(z) &= \sum_{n=1} nt_n z^{n-1}, \hspace{3cm} \tilde{S}_{+} (z)
= -\sum_{n=1}^{\infty} v_{-n} z^{n-1},\\
S_-(z) &=-t_0 z^{-1} -\sum_{n=1}^{\infty}
v_nz^{-n-1},\hspace{1.5cm} \tilde{S}_{-}(z) = -t_0z^{-1}
+\sum_{n=1}^{\infty} nt_{-n} z^{-n-1}.
\end{align*}

 Let
$F$ and $G$ be the inverse functions of $f$ and $g$ respectively.
Using Lemma \ref{lemma3}, which still holds for $(f,g)\in
\mathfrak{D}$, we can prove as in Propositions \ref{Prop4} and
\ref{Prop5} that
\begin{proposition}
Let $\kappa_{n,m}$ be the generalized Grunsky coefficients of the
pair $(F, G)$. The variation of the function $v_m$ with
respect to the coordinates $t_n$, $n \in \Z$ is given by the
following. For $m\neq 0$,
\begin{align*}
\frac{\pa v_m}{\pa t_n} &= -|mn|\kappa_{n,m} , \hspace{0.5cm}
n\neq 0,\hspace{1cm}\frac{\pa v_m}{\pa t_0} = |m|\kappa_{0,m},
\end{align*}
and for $m=0$,
\begin{align*}
\frac{\pa v_0}{\pa t_n}&= |n|\kappa_{n,0},\hspace{0.5cm} n\neq
0,\hspace{1cm} \frac{\pa v_0}{\pa t_0} = 2\log b,
\end{align*}
and
\begin{align*}
\frac{\pa v_m}{\pa \bar{t}_n}&=0, \hspace{1cm}\text{for all}\;
m,n.
\end{align*}
\end{proposition}
\begin{proof}
We observe that
\begin{align*}
\frac{\pa S_{\pm}}{\pa t} (z)=& \frac{1}{2\pi i} \oint_{S^1}
\Biggl(\frac{ \left(\pa(1/f)/\pa t\right)(w) g'(w)}{g(w) -z}
+\frac{(1/f)(w) (\pa^2 g(w)/\pa t\pa w)}{g(w)-z} \\
&\hspace{5cm}-\frac{(1/f)(w) g'(w)( \pa g/\pa
t)(w)}{(g(w)-z)^2} \Biggr)dw\\
=& \frac{1}{2\pi i} \oint_{S^1} \left(\frac{ (\pa (1/f)/\pa t)(w)
g'(w)}{g(w) -z} +\left(\frac{1}{f}\right)(w)\frac{\pa }{\pa w}
\left(\frac{\pa g(w)/\pa t}{g(w)-z}\right)\right)dw\\
=&\frac{1}{2\pi i} \oint_{S^1} \left(\frac{(\pa (1/f)/\pa
t)(w)g'(w)}{g(w) -z} - \frac{(1/f)'(w)(\pa
g(w)/\pa t)}{g(w)-z}\right)dw\\
=&\frac{1}{2\pi i} \oint_{S^1} \frac{ \left(\left(\pa
\bigl((1/f)\circ G\bigr)\right)/\pa t)\bigr)\right)\circ g(w)
g'(w)}{g(w)-z}dw.
\end{align*}
Similarly, we can show that
\begin{align*}
\frac{\pa \tilde{S}_{\pm}}{\pa t}(z) =\frac{1}{2\pi i} \oint_{S^1}
\frac{\left(\pa (g\circ F)/\pa t\right)\circ f(w) f'(w)}{(f(w))^2(f(w) -z)} dw.
\end{align*}
The rest of the proof is the same.
\end{proof}

We define the functions $\Phi$ and $\Psi$ as in Section 5. Now we
define the tau function $\tau$ to be the real valued function
given by
\begin{align*}
4\log \tau = \Bigl(2t_0 v_0 -t_0^2 &+\frac{1}{2\pi i}\oint_{S^1}
\frac{g'(w)}{f(w)}\left(g(w)\Phi'(g(w))+
2\Phi(g(w))\right) dw \\
&+\frac{1}{2\pi i}\oint_{S^1}
\frac{g(w)f'(w)}{(f(w))^2}\left(f(w)\Psi'(f(w)) -
2\Psi(f(w))\right) dw\Bigr) \\
&+ \text{complex conjugate}.
\end{align*}
Then we have
\begin{proposition}For all $n$, we have
\begin{align*}
\frac{\pa \log \tau}{\pa t_n}=v_n, \hspace{3cm}\frac{\pa \log
\tau}{\pa \bar{t}_n}=\bar{v}_n.
\end{align*}
\end{proposition}

$\log \tau$ is a harmonic function on $\mathfrak{D}$. It generates
two solutions of the dispersionless Toda hierarchy. Namely if we
take $t_n$, $n\in \Z$ as the time variables, we find that $(\mL=g,
\mL=f)$ is a solution to the  dispersionless Toda hierarchy.
 Moreover, all the variations of the coefficients of $f
,g$ with respect to $\bar{t}_n$ vanish. On the other hand, if we
take $\bar{t}_n$ as the time variables, then $(\mL=\bar{g},
\tilde{\mL}=\bar{f})$ is a solution to the dispersionless Toda
hierarchy. We have introduced a complex structure on the space
$\mathfrak{D}$ such that the coefficients of $f$ and $g$ are
holomorphic functions under this complex structure.

On the subspace of $\mathfrak{D}$ defined by $f(z)
=1/\bar{g}(z^{-1})$, one finds that $b=g'(\infty)$ is real,
$t_{-n} = -\bar{t}_n$, $v_{-n}=-\bar{v}_n$, for all $n\neq 0$, and
$t_0$, $v_0$ are real. Hence $\{t_n, \bar{t}_n, n\geq 1, t_0\}$
form a set of local coordinates on this space, which is precisely
the space considered by Wiegmann and Zabrodin. It is easy to check
that the definitions of $t_n, v_n$, $n\geq 0$ coincide with the
definition of Wiegmann and Zabrodin. However, we do not say that
all the evolution equations on the space $\mathfrak{D}$ can be
carried over to its subspace. The vector fields $\pa/\pa t_n$ have
different meanings.

The subspace of $\mathfrak{D}$ defined by $f(S^1)=g(S^1)$
corresponds to the conformal welding problem we discuss in Section
5.  Hence we see that the solution of Wiegmann and Zabrodin and
our solution to the conformal welding problem actually correspond
to two different subsystem of the same system.

 \vspace{0.2cm} \noindent \textbf{Acknowledgement}\;
The author would like to thank the Ministry of Science, Technology
and Innovation of Malaysia for funding this project under
eScienceFund 06-02-01-SF0021.

\appendix
\section{The subgroup of linear fractional transformations}

Here we consider the simplest case of domains, namely discs. The
most general conformal maps $f$ and $g$ mapping $\mathbb{D}$ and
$\mathbb{D}^*$ to the interior and exterior of a circle and
satisfying our normalization conditions are given by
\begin{align*}
f(z) = \frac{1}{b}\frac{z}{1+az}, \hspace{0.5cm} |a|<1,\hspace{2cm} g(z) = b
z+c,\hspace{0.5cm} b\neq 0.
\end{align*}
The necessary condition that $f$ and $g$ map $S^1$ to the same
circle is that $\gamma=g^{-1}\circ f$ is a linear fractional
transformation on the unit circle, i.e. it belongs to
$\PSL(2,\R)$. This implies that
\begin{align*}
b = \frac{e^{\frac{i\alpha}{2}}}{\sqrt{ (1-|a|^2)}},  \;\;c= -\frac{\bar{a}e^{-\frac{i\alpha}{2}}}{\sqrt{ (1-|a|^2)}},
\hspace{0.5cm}\text{and}\hspace{0.5cm} \gamma(z)=
e^{-i\alpha}\frac{z+\bar{a}}{1+az}.
\end{align*}
For the inverse mappings $(g^{-1}, f^{-1})$, we consider the Fourier expansions:
\begin{align*}
\gamma(w) = \bar{a}e^{-i\alpha} +\sum_{n=0}^{\infty}(-1)^n
e^{-i\alpha}
a^{n} (1-|a|^2) w^{n+1},\\
 \frac{1}{\gamma^{-1}}(w) =- a+ \sum_{n=0}^{\infty}
 e^{-i(n+1)\alpha}(1-|a|^2)\bar{a}^nw^{-n-1}.
\end{align*}
The variables $t_n$ and $v_n$ are given by
\begin{align*}
t_{1}&=- a, \;t_0= e^{-i\alpha}(1-|a|^2),\;  t_{-1}=-\bar{a}
e^{-i\alpha},\;\hspace{0.5cm} t_n=0, \hspace{0.5cm}\forall |n|\geq
2,\\
v_n &=
 e^{-i(n+1)\alpha}(1-|a|^2)\bar{a}^n, \hspace{0.3cm} v_{-n}=(-1)^{n-1}e^{-i\alpha}
a^{n} (1-|a|^2), \hspace{0.5cm}n\geq 1.
\end{align*}
Hence
\begin{align*}
v_n&=(-1)^n t_0t_{-1}^n, \hspace{0.5cm} v_{-n}= -t_0t_1^n, \hspace{0.5cm}
 v_0 = t_0 \log t_0 -t_0-t_1t_{-1},
\end{align*}
and the tau function is given by
\begin{align*}
\log \tau = \frac{t_0^2}{2} \log t_0 -\frac{3}{4}
t_0^2-t_0t_1t_{-1}.
\end{align*}
The restriction $b$ being real amount to restricting $\alpha=0$.
In that case $t_0$ depend on $t_1$ and $t_{-1}$ by
$t_0=1-t_1t_{-1}$.

For the mappings $(g, f)$,
\begin{align*}
t_n =\frac{1}{2\pi in} \oint_{S^1}\frac{g(w)^{-n}}{f(w)}dg(w)
=\begin{cases}ab=
\frac{ae^{\frac{i\alpha}{2}}}{\sqrt{1-|a|^2}},\hspace{0.5cm}& \text{if}\;n=1,\\
0,\hspace{0.5cm}& \text{if}\;n\geq 2.
\end{cases}
\end{align*}
\begin{align*}
t_{-n} =\frac{-1}{2\pi in} \oint_{S^1}f(w)^{n-2}g(w)df(w)
=\begin{cases}
-c=\frac{\bar{a}e^{-\frac{i\alpha}{2}}}{\sqrt{1-|a|^2}}\hspace{0.5cm}& \text{if}\;n=1,\\
0\hspace{0.5cm}& \text{if}\;n\neq1.
\end{cases}
\end{align*}
Similarly, we find that
\begin{align*}
t_0 = b^2=& e^{i\alpha}(1-|a|^2)^{-1}, \\ v_n = b^2c^n,&\hspace{1cm}
v_{-n} = -a^n b^{n+2},\hspace{1cm}n\geq 1.
\end{align*}
Hence, as functions of $t_n$, we still have
\begin{align*}
v_n&=(-1)^n t_0t_{-1}^n, \hspace{0.5cm} v_{-1}= -t_0t_1^n, \hspace{0.5cm}
 v_0 = t_0 \log t_0 -t_0-t_1t_{-1},
\end{align*}
and the tau function is given by
\begin{align*}
\log \tau = \frac{t_0^2}{2} \log t_0 -\frac{3}{4}
t_0^2-t_0t_1t_{-1}.
\end{align*}

\end{document}